\documentclass[sigconf,nonacm,balance=false]{acmart}

\usepackage{amsmath}  
\usepackage{algorithm}
\usepackage{algpseudocode}
\usepackage{booktabs}
\usepackage[capitalize,noabbrev]{cleveref}
\usepackage{enumitem}
\usepackage{flushend}
\usepackage{graphicx}
\usepackage{multirow}
\usepackage{tikz}
\usepackage{xcolor}
\usepackage{xspace}
\usetikzlibrary{arrows.meta,calc,patterns,decorations.pathreplacing}

\definecolor{tdSM}{HTML}{E69F00}
\definecolor{tdSMd}{HTML}{B07800}
\definecolor{tdSMf}{HTML}{FCE6BE}
\definecolor{tdCE}{HTML}{0072B2}
\definecolor{tdCEf}{HTML}{D6E8F5}
\definecolor{tdNET}{HTML}{009E73}
\definecolor{tdNETf}{HTML}{D5F0E7}
\definecolor{tdNVS}{HTML}{AA3377}
\definecolor{tdNVSf}{HTML}{F1DEE9}
\definecolor{tdTAX}{HTML}{D55E00}
\definecolor{tdTAXf}{HTML}{F8D3BD}
\definecolor{tdINK}{HTML}{2F2F2F}
\definecolor{tdMUTE}{HTML}{6B6B6B}
\definecolor{tdRULE}{HTML}{BDBDBD}
\definecolor{tdCMP}{HTML}{C4C4C4}
\definecolor{tdCMPf}{HTML}{E8E8E8}
\definecolor{tdID0}{HTML}{F9E39A}
\definecolor{tdID1}{HTML}{A6D8CF}
\definecolor{tdID2}{HTML}{E8968A}
\definecolor{tdID3}{HTML}{7F86AE}
\definecolor{tdID4}{HTML}{BCCB85}

\newcommand{\tdfont}{\fontsize{7}{8}\selectfont\sffamily}
\newcommand{\tdsmall}{\fontsize{7}{8}\selectfont\sffamily}
\newcommand{\tdtiny}{\fontsize{7}{8}\selectfont\sffamily}
\newcommand{\tdbold}{\fontsize{7}{8}\selectfont\sffamily\bfseries}
\DeclareRobustCommand{\tdref}[1]{\tikz[baseline=(tdn.base)]\node[tdstep](tdn){#1};}

\tikzset{
  td/.style={font=\tdfont, line width=0.45pt, line cap=round, line join=round,
    >={Latex[length=2.6pt,width=2.2pt]}},
  tdtile/.style={draw=tdINK!50, line width=0.25pt},
  tdtax/.style={draw=tdTAX, line width=0.6pt, dash pattern=on 1.5pt off 1pt},
  tdbox/.style={draw=tdINK, line width=0.45pt, rounded corners=1.2pt},
  tdce/.style={draw=tdCE, line width=0.7pt},
  tdnvs/.style={draw=tdNVS, line width=0.7pt},
  tdnet/.style={draw=tdNET, line width=0.85pt, dash pattern=on 2.2pt off 1.1pt},
  tdsm/.style={draw=tdSMd, line width=0.6pt},
  tdstep/.style={circle, draw=tdINK, fill=white, line width=0.4pt, inner sep=0pt,
    minimum size=8.8pt, font=\fontsize{7}{7}\selectfont\sffamily\bfseries},
}

\setlist[itemize]{labelindent=0.3em,leftmargin=*,labelsep=0.4em,topsep=2.5pt,itemsep=1.5pt,parsep=0pt}
\setlist[enumerate]{label=(\arabic*),labelindent=0.3em,leftmargin=*,labelsep=0.4em,topsep=2.5pt,itemsep=1.5pt,parsep=0pt}
\graphicspath{{figures/}}
\makeatletter
\renewcommand{\@authorfont}{\large}
\renewcommand{\@affiliationfont}{\small}
\makeatother
\renewcommand\footnotetextcopyrightpermission[1]{}
\AtBeginDocument{}

\newcommand{\system}{\textsc{TIDE}\xspace}
\newcommand{\code}[1]{\ifmmode\mathtt{#1}\else{\normalfont\ttfamily #1}\fi}
\newcommand{\dtype}[1]{\ifmmode\mathsf{#1}\else{\normalfont\sffamily #1}\fi}

\begin{document}

\title[The Endgame of FSDP Communication]{
The Endgame of FSDP Communication:
Layout-Aware Collectives on Copy Engines, RDMA, and NVSwitch}

\author{Zejian Wang}
\authornote{This work was conducted during the author's internship at KlingAI.}
\email{wangzejian@kuaishou.com}
\affiliation{%
  \institution{KlingAI Research}
  \city{Shanghai}
  \country{China}
}

\author{Yifan Zhao}
\email{zhaoyifan14@kuaishou.com}
\affiliation{%
  \institution{KlingAI Research}
  \city{Shanghai}
  \country{China}
}

\author{Kaihui Ren}
\authornotemark[1]
\email{renkaihui@kuaishou.com}
\affiliation{%
  \institution{KlingAI Research}
  \city{Beijing}
  \country{China}
}

\author{Xiaoling Zhou}
\authornotemark[1]
\email{zhouxiaoling03@kuaishou.com}
\affiliation{%
  \institution{KlingAI Research}
  \city{Beijing}
  \country{China}
}

\author{Yubin Wang}
\email{wangyubin05@kuaishou.com}
\affiliation{%
  \institution{KlingAI Research}
  \city{Beijing}
  \country{China}
}

\author{Chenhao Guan}
\authornotemark[1]
\email{guanchenhao@kuaishou.com}
\affiliation{%
  \institution{KlingAI Research}
  \city{Shanghai}
  \country{China}
}

\author{Lin Zhang}
\email{zhanglin22@kuaishou.com}
\affiliation{%
  \institution{KlingAI Research}
  \city{Shanghai}
  \country{China}
}

\author{Yufeng Li}
\email{liyufeng06@kuaishou.com}
\affiliation{%
  \institution{KlingAI Research}
  \city{Shanghai}
  \country{China}
}

\author{Zhikang Mo}
\email{mozhikang@kuaishou.com}
\affiliation{%
  \institution{KlingAI Research}
  \city{Shanghai}
  \country{China}
}

\author{Zhuoru Ma}
\email{mazhuoru@kuaishou.com}
\affiliation{%
  \institution{KlingAI Research}
  \city{Shanghai}
  \country{China}
}

\author{Yingtong Xiong}
\email{xiongyingtong@kuaishou.com}
\affiliation{%
  \institution{KlingAI Research}
  \city{Shanghai}
  \country{China}
}

\author{Jihu Guo}
\authornotemark[1]
\email{guojihu@kuaishou.com}
\affiliation{%
  \institution{KlingAI Research}
  \city{Shanghai}
  \country{China}
}

\author{Yiqiao Liao}
\authornote{Corresponding authors.}
\email{liaoyiqiao@kling.ai}
\affiliation{%
  \institution{KlingAI Research}
  \city{Beijing}
  \country{China}
}

\author{Chengru Song}
\authornotemark[2]
\email{songchengru@kuaishou.com}
\affiliation{%
  \institution{KlingAI Research}
  \city{Beijing}
  \country{China}
}

\renewcommand{\shortauthors}{Zejian Wang et al.}

\begin{abstract}
Fully Sharded Data Parallelism (FSDP2) maintains parameters and
gradients as per-parameter tensors, whereas its collective path
operates on flat, rank-major buffers. Converting between these layouts
copies all-gather outputs into parameter storage and packs gradients
into full-group \dtype{fp32} buffers. Layout conversion alone accounts
for 6.4\% of step time in Llama-3.1 70B training at 8K tokens per GPU.
Communication kernels impose a second cost: a default NCCL
reduce-scatter can slow a concurrent GEMM that would otherwise
complete in a single wave by 69\%.

We present \system, a set of zero-copy collectives built on public
NCCL~2.31+ APIs. \system preserves FSDP2's per-parameter storage and
compute kernels while folding layout placement into collective
addressing. \system-AG delivers shards directly to parameter storage
using NVSwitch multicast and paced recursive-doubling RDMA, with no
SMs used for transport or placement. \system-RS consumes per-parameter
\dtype{bf16} gradients directly and reduces them into \dtype{fp32}
shards---a mixed-precision mode that NCCL does not provide---using copy
engines and short-lived thread blocks (CTAs) for node-local accumulation
and NVSwitch for cross-node reduction. A two-phase pipeline separates gradient reads from network
progress; Slot-Reuse optionally stages gradients in vacated all-gather
slots.

On 64 GPUs, \system reduces step time by up to \textbf{27.1\%}
relative to Stock FSDP2 and \textbf{6.0\%} relative to
FSDP2+SymmMem across Llama-3.1 70B, Wan2.2, and Qwen3-30B-A3B. It
reduces reduce-scatter SM time by factors of $12.5$--$14.5$ and cuts
interference with concurrent compute by up to an order of magnitude.
Its reduce-scatter buffer remains nearly constant from 16 to 64 GPUs
at a fixed per-GPU output size, and its peak memory is up to
4.7~GiB below FSDP2+SymmMem's and, with Slot-Reuse on Wan2.2, below
Stock FSDP2's.
\end{abstract}


\maketitle

\section{Introduction}
\label{sec:intro}

Fully Sharded Data Parallelism
(FSDP2~\cite{liang2025torchtitan,zhao2023pytorchfsdp}), rooted in
ZeRO-3~\cite{rajbhandari2020zero}, is the workhorse paradigm for training
dense and Mixture-of-Experts (MoE) foundation
models~\cite{fedus2022switch,lepikhin2020gshard,rajbhandari2022deepspeedmoe}
across hundreds to thousands of GPUs~\cite{jiang2024megascale,qian2024alibaba}.
FSDP2 shards parameters, gradients, and optimizer states across $N$
data-parallel ranks, rebuilds each layer's parameters with an
\emph{all-gather} (AG), and aggregates its gradients with a
\emph{reduce-scatter} (RS).

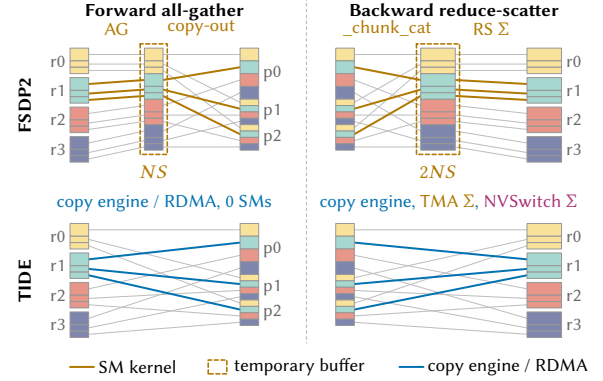
\begin{figure}[t]
  \centering
\begin{tikzpicture}[td, x=1pt, y=1pt,
  tdlinkbg/.style={draw=tdRULE, line width=0.3pt},
  tdbuf/.style={draw=tdSMd, line width=0.6pt, dash pattern=on 1.5pt off 1pt}]
\def\tdU{2.4}\def\tdG{1.5}\def\tdWb{7}\def\tdWf{13}
\def\tdXa{22}\def\tdXb{50}\def\tdXc{86}\def\tdXd{122}\def\tdXe{154}\def\tdXf{194}
\def\tdYf{-20}\def\tdYt{-86}
\def\tdC#1#2#3#4#5{%
  \pgfmathsetmacro\tds{{2,1,1}[#4]}%
  \ifnum#1=0 \pgfmathsetmacro\tdo{#5*4+{0,2,3}[#4]}%
  \else \pgfmathsetmacro\tdo{{0,8,12}[#4]+#5*\tds}\fi
  \pgfmathsetmacro\tdy{#3-(\tdo+0.5*\tds)*\tdU-#2*#5*\tdG}}
\def\tdCol#1#2#3#4#5{
  \foreach \tdp in {0,1,2}{\foreach \tdr in {0,1,2,3}{%
    \tdC{#4}{#5}{#2}{\tdp}{\tdr}%
    \filldraw[tdtile, fill=tdID\tdr] (#1,\tdy+0.5*\tds*\tdU) rectangle ++(#3,-\tds*\tdU);}}}
\def\tdLink#1#2#3#4#5#6#7#8#9{
  \foreach \tdp in {0,1,2}{\foreach \tdr in {0,2,3}{%
    \tdC{#3}{#4}{#2}{\tdp}{\tdr}\let\tdyA\tdy
    \tdC{#7}{#8}{#6}{\tdp}{\tdr}\let\tdyB\tdy
    \draw[tdlinkbg] (#1,\tdyA) -- (#5,\tdyB);}}
  \foreach \tdp in {0,1,2}{%
    \tdC{#3}{#4}{#2}{\tdp}{1}\let\tdyA\tdy
    \tdC{#7}{#8}{#6}{\tdp}{1}\let\tdyB\tdy
    \draw[#9, line width=0.75pt] (#1,\tdyA) -- (#5,\tdyB);}}
\def\tdRanks#1#2#3{
  \foreach \tdr in {0,1,2,3}{%
    \node[font=\tdtiny, text=tdMUTE, anchor=#3, inner sep=1pt]
      at (#1,#2-\tdr*4*\tdU-2*\tdU-\tdr*\tdG) {r\tdr};}}
\def\tdParams#1#2{
  \foreach \tdp/\tdm in {0/4,1/10,2/14}{%
    \node[font=\tdtiny, text=tdMUTE, anchor=west, inner sep=1pt]
      at (#1,#2-\tdm*\tdU) {p\tdp};}}

\node[font=\tdbold, anchor=base] at (57.5,-8) {Forward all-gather};
\node[font=\tdbold, anchor=base] at (166.5,-8) {Backward reduce-scatter};
\draw[tdRULE, dash pattern=on 1pt off 1.2pt] (111,-1) -- (111,-131);
\node[font=\tdbold, rotate=90] at (5,\tdYf-21.5) {FSDP2};
\node[font=\tdbold, rotate=90] at (5,\tdYt-21.5) {TIDE};

\tdLink{\tdXa+\tdWb}{\tdYf}{0}{1}{\tdXb}{\tdYf}{0}{0}{draw=tdSMd}
\tdLink{\tdXb+\tdWb}{\tdYf}{0}{0}{\tdXc}{\tdYf}{1}{0}{draw=tdSMd}
\tdLink{\tdXd+\tdWb}{\tdYf}{1}{0}{\tdXe}{\tdYf}{0}{0}{draw=tdSMd}
\tdLink{\tdXe+\tdWf}{\tdYf}{0}{0}{\tdXf}{\tdYf}{0}{1}{draw=tdSMd}
\tdCol{\tdXa}{\tdYf}{\tdWb}{0}{1}
\tdCol{\tdXb}{\tdYf}{\tdWb}{0}{0}
\tdCol{\tdXc}{\tdYf}{\tdWb}{1}{0}
\tdCol{\tdXd}{\tdYf}{\tdWb}{1}{0}
\tdCol{\tdXe}{\tdYf}{\tdWf}{0}{0}
\tdCol{\tdXf}{\tdYf}{\tdWf}{0}{1}
\draw[tdbuf] (\tdXb-1.6,\tdYf+1.6) rectangle (\tdXb+\tdWb+1.6,\tdYf-16*\tdU-1.6);
\draw[tdbuf] (\tdXe-1.6,\tdYf+1.6) rectangle (\tdXe+\tdWf+1.6,\tdYf-16*\tdU-1.6);
\node[font=\tdsmall, text=tdSMd, anchor=north] at (\tdXb+0.5*\tdWb,\tdYf-16*\tdU-2.5) {$NS$};
\node[font=\tdsmall, text=tdSMd, anchor=north] at (\tdXe+0.5*\tdWf,\tdYf-16*\tdU-2.5) {$2NS$};
\node[font=\tdsmall, text=tdSMd, anchor=south] at (39.5,\tdYf+1.5) {AG};
\node[font=\tdsmall, text=tdSMd, anchor=south] at (71.5,\tdYf+1.5) {copy-out};
\node[font=\tdsmall, text=tdSMd, anchor=south] at (141.5,\tdYf+1.5) {\_chunk\_cat};
\node[font=\tdsmall, text=tdSMd, anchor=south] at (180.5,\tdYf+1.5) {RS $\Sigma$};
\tdRanks{\tdXa-0.5}{\tdYf}{east}
\tdRanks{\tdXf+\tdWf+0.5}{\tdYf}{west}
\tdParams{\tdXc+\tdWb+0.5}{\tdYf}

\tdLink{\tdXa+\tdWb}{\tdYt}{0}{1}{\tdXc}{\tdYt}{1}{0}{draw=tdCE}
\tdLink{\tdXd+\tdWb}{\tdYt}{1}{0}{\tdXf}{\tdYt}{0}{1}{draw=tdCE}
\tdCol{\tdXa}{\tdYt}{\tdWb}{0}{1}
\tdCol{\tdXc}{\tdYt}{\tdWb}{1}{0}
\tdCol{\tdXd}{\tdYt}{\tdWb}{1}{0}
\tdCol{\tdXf}{\tdYt}{\tdWf}{0}{1}
\node[font=\tdsmall, text=tdCE, anchor=south] at (57.5,\tdYt+1.5) {copy engine / RDMA, 0 SMs};
\node[font=\tdsmall, text=tdCE, anchor=south] at (164.5,\tdYt+1.5)
  {copy engine, \textcolor{tdSMd}{TMA $\Sigma$}, \textcolor{tdNVS}{NVSwitch $\Sigma$}};
\tdRanks{\tdXa-0.5}{\tdYt}{east}
\tdRanks{\tdXf+\tdWf+0.5}{\tdYt}{west}
\tdParams{\tdXc+\tdWb+0.5}{\tdYt}

\begin{scope}[shift={(22,-140)}, font=\tdsmall]
  \draw[draw=tdSMd, line width=0.75pt] (0,0) -- (9,0);
  \node[anchor=west, inner sep=1.5pt] at (9,0) {SM kernel};
  \draw[tdbuf] (52,-2.6) rectangle ++(7,5.2);
  \node[anchor=west, inner sep=1.5pt] at (60,0) {temporary buffer};
  \draw[draw=tdCE, line width=0.75pt] (124,0) -- (133,0);
  \node[anchor=west, inner sep=1.5pt] at (133,0) {copy engine / RDMA};
\end{scope}
\end{tikzpicture}
  \caption{FSDP2's contract boundary for one parameter group (3
  parameters, $N=4$ ranks, shard size $S$; tile color marks the owning
  rank, and colored lines trace rank~1's pieces). Stock FSDP2 converts layouts with SM
  kernels and temporary buffers; \system writes each piece once, at its
  final address.}
  \Description{Two rows of tensor-layout diagrams comparing stock FSDP2
  and TIDE for all-gather and reduce-scatter.}
  \label{fig:boundary-tax}
\end{figure}

\begin{table*}[t]
  \centering
  \caption{Design space of fully sharded communication subsystems.}
  \label{tab:design-space}
  \small
  \setlength{\tabcolsep}{3pt}
  \begin{tabular*}{\textwidth}{@{\extracolsep{\fill}}lccccc@{}}
    \toprule
    \textbf{Dimension} & \textbf{Stock FSDP2} & \textbf{FSDP2+SymmMem}~\cite{pytorchsymmmem2025} & \textbf{Flat-Buffer FSDP}~\cite{megatronfsdp2025,vescalefsdp2025} & \textbf{PACE}~\cite{pace2025} & \textbf{\system (Ours)} \\
    \midrule
    Preserves \code{DTensor} storage & Yes & Yes & \textbf{No} (flat/uneven) & Yes & \textbf{Yes} \\
    AG direct landing & No & No & Yes & Yes & \textbf{Yes}\textsuperscript{$\ddagger$} \\
    RS input preparation & Full-group pack & Full-group pack & Fused wgrad / copy & SM staging / cast & \textbf{In-Place Read}\textsuperscript{$\dagger$} \\
    \addlinespace[2pt]
    AG SM budget & 16--32 & 0  & / & 0--1 & \textbf{0} \\
    \addlinespace[2pt]
    RS SM budget & 16--32 & $\approx 13$ & / & 32 & \textbf{8}\textsuperscript{\S} \\
    Extra buffer footprint & $3NS$ & $3NS$ private & $\Theta(NS)$ & $2NS + \text{rings}$ & $\mathbf{\approx 1.1NS}$ \\
    \bottomrule
  \end{tabular*}\par\vspace{3pt}
  {\footnotesize\raggedright
  ``/'' inherits the configured NCCL CTA budget. Buffer
  footprints are large-scale \dtype{bf16}/\dtype{fp32} terms, with $NS$
  denoting one unsharded \dtype{bf16} group; see
  \Cref{sec:memory-model}.\par
  \textsuperscript{$\ddagger$}Mixed-precision AG retains
  the \dtype{fp32}$\to$\dtype{bf16} cast.
  \textsuperscript{$\dagger$}Slot-Reuse copies gradients into the vacated
  AG slot. \textsuperscript{\S}\system uses short-lived sum and landing
  CTAs with tunable budgets; see
  \Cref{sec:wave-quant,sec:solver}.\par}
\end{table*}

\paragraph{Wire speed is not end-to-end efficiency.}
Well-tuned large-message collectives approach link bandwidth on
NVSwitch nodes with rail-aligned RoCEv2 or InfiniBand NICs, yet FSDP2
pays for more than wire time. Its model view is per-parameter: $P$
independently shaped tensors, each with a persistent dimension-0
\code{DTensor} shard~\cite{ansel2024pytorch2,paszke2019pytorch}. Its
collective path is rank-major and flat: NCCL sees contiguous buffers
ordered by rank, and our target configurations reduce gradients in
\dtype{fp32}. Translating between the two layouts and sharing SMs and
HBM with the model cost three measurable forms of overhead
(\Cref{fig:boundary-tax}, \Cref{sec:principles}):
\begin{enumerate}
  \item \textbf{Layout conversion and packing.} Stock FSDP2 copies each
  rank-major all-gather output into per-parameter storage and packs
  \dtype{bf16} gradients into a transient full-group \dtype{fp32}
  buffer before each reduce-scatter, so each group's parameters and
  gradients briefly exist twice, and the copies take 0.33~s of each
  Llama-3.1 70B step at 8K tokens/GPU.
  \item \textbf{SM wave quantization.} Small-$M$ GEMMs in MoE layers and
  narrow projections often finish in one wave of thread-block clusters.
  A reduce-scatter that stays resident for the whole transfer pushes
  part of that wave into a second one: NCCL's default
  reduce-scatter slows such a GEMM by $69\%$, and its
  symmetric-memory variant by $11\%$ at the median.
  \item \textbf{HBM direction contention.} Per byte served, remote
  service reads and switch-reduction reads slow co-running
  memory-bound kernels $2.2$--$4\times$ more than inbound copy-engine
  writes, depending on the clock regime (\Cref{sec:direction-tax}).
\end{enumerate}

This overhead scales with sharded parameter bytes, whereas useful
compute scales with activated parameters times tokens per GPU:
\begin{equation}
  \text{Recoverable FSDP Overhead Fraction} \;\propto\;
  \frac{\Psi_{\mathrm{total}}}{\Psi_{\mathrm{active}} \cdot T_{\mathrm{tok/gpu}}},
  \label{eq:tax-scaling}
\end{equation}
where $\Psi_{\mathrm{total}}$ is the sharded parameter volume,
$\Psi_{\mathrm{active}}$ is the activated parameter count per token, and
$T_{\mathrm{tok/gpu}}$ is the per-GPU token budget per step. The overhead
therefore grows under strong scaling, sequence parallelism, and sparse
MoE routing ($\Psi_{\mathrm{total}} \gg \Psi_{\mathrm{active}}$). On
dense Llama-3.1 70B, \system's absolute saving over Stock FSDP2 stays at
0.51--0.54~s from 16K to 32K tokens/GPU, so its relative reduction falls
from 6.1\% to 3.1\% (\Cref{fig:eval-e2e}a).

\paragraph{Why existing approaches fall short.}
Existing systems fix one side of the problem or change FSDP2's storage
contract (\Cref{tab:design-space}). PyTorch's symmetric-memory FSDP2
path~\cite{nccl231,pytorchsymmmem2025} removes all-gather transport SMs
but keeps rank-major outputs, full-group gradient packing, and 13
reduce-scatter CTAs. AMD MORI~\cite{mori2025} lands all-gather data
directly within one node; PACE~\cite{pace2025} extends direct placement
to reduce-scatter but stages and converts its multi-tensor input in a
32-CTA SM ring. Megatron-FSDP~\cite{megatronfsdp2025} and
veScale-FSDP~\cite{vescalefsdp2025} avoid packing by replacing
per-parameter storage with persistent flat windows. Among the systems
considered, none consumes per-parameter tensors directly across nodes, keeps
\dtype{fp32} reduction, and moves transport off SMs while leaving
FSDP2's storage intact.

\paragraph{Our approach and contributions.}
\textbf{Layout conversion in FSDP2 is a collective problem, not a copy
problem.} When the collective itself reads every tensor where the model
stores it and delivers every result to its final address---with copy
engines, RDMA, and NVSwitch moving the bytes and SMs performing only
short-lived \dtype{fp32} accumulation---the layout copies, the packing
buffers, and most of the SM interference disappear together. \system
builds such collectives on public NCCL~2.31+ APIs
(\Cref{fig:arch-overview}). They are
zero-copy in a precise sense: FSDP2's layout copies (all-gather copy-in
and copy-out, and reduce-scatter packing) and their full-group buffers
are gone; only required dtype casts, bounded multi-node staging, and an
optional memory-saving gradient copy remain. The design also frees
memory: all-gather slots double as parameter storage and, optionally, as
reduce-scatter input. This paper makes three contributions:
\begin{itemize}
  \item \textbf{Zero-copy collectives
  (\Cref{sec:principles,sec:ag,sec:rs}):} \system-AG and \system-RS
  consume FSDP2's per-parameter tensors without packing. \system-AG uses
  no SMs: copy engines multicast shards through NVSwitch, and paced
  recursive-doubling RDMA carries them across nodes. \system-RS reduces
  \dtype{bf16} gradients into \dtype{fp32} shards, a mixed-precision mode
  that NCCL lacks. Copy engines push slices to their owners, short-lived
  TMA CTAs sum them into \dtype{fp32} node partials, and Latin-square
  placement lets the destination NVSwitch sum up to $L$ nodes' partials
  in one reduction.
  \item \textbf{Pipelining, memory, and runtime
  (\Cref{sec:pipeline,sec:runtime}):} A two-phase pipeline separates
  gradient reads from network progress. In-Place mode reads gradients
  where autograd leaves them and runs the next all-gather early;
  Slot-Reuse mode copies them into a vacated all-gather slot and frees
  their storage. For 64 nodes, the analysis in \Cref{sec:memory-model}
  bounds the live bytes during a layer's backward at $3.72NS$
  (In-Place) and $2.29NS$ (Slot-Reuse), versus FSDP2's $4NS$; \system's
  measured peak memory (\Cref{sec:eval-e2e}) is up to 4.7~GiB below
  FSDP2+SymmMem's and, with Slot-Reuse, below Stock FSDP2's on Wan2.2.
  A topology-scoped rail communicator keeps connection, memory,
  submission, and SM costs bounded at scale.
  \item \textbf{Evaluation (\Cref{sec:eval}):} On 64 GPUs, \system is
  the fastest system at every measured point across Llama-3.1 70B,
  Wan2.2, and Qwen3-30B-A3B, reducing step time by up to 27.1\%
  relative to Stock FSDP2 and 6.0\% relative to
  FSDP2+SymmMem. Training-shape replays attribute the gains to zero-SM
  all-gather and $12.5$--$14.5\times$ less reduce-scatter SM time.
\end{itemize}

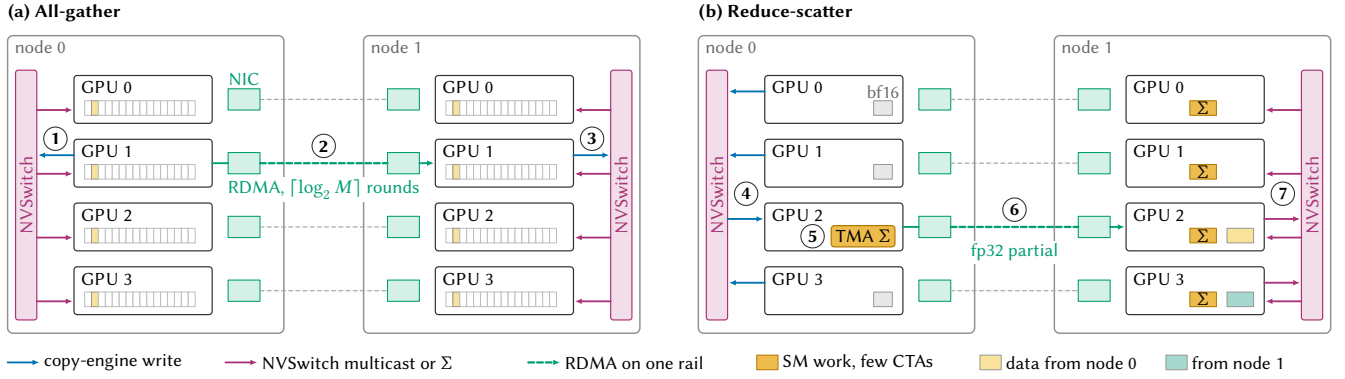
\begin{figure*}[t]
  \centering
\begin{tikzpicture}[td, x=1pt, y=1pt]
\def\tdTop#1{\pgfmathsetmacro\tdy{-28-24*#1}}
\def\tdScaffold#1#2{%
  \begin{scope}[shift={(#1,0)}]
  \foreach \tdx/\tdlab in {0/node 0, 134/node 1}{%
    \draw[draw=tdMUTE!75, line width=0.55pt, rounded corners=2pt] (\tdx,-13) rectangle ++(104,-112);
    \node[font=\tdsmall, text=tdMUTE, anchor=north west, inner sep=1.5pt] at (\tdx+1,-13) {\tdlab};}
  \foreach \tdx in {3,227}{%
    \filldraw[draw=tdNVS, fill=tdNVSf, line width=0.45pt, rounded corners=1pt] (\tdx,-26) rectangle ++(8,-94);
    \node[font=\tdtiny, text=tdNVS, rotate=90] at (\tdx+4,-73) {NVSwitch};}
  \foreach \tdj in {0,...,3}{%
    \tdTop{\tdj}%
    \foreach \tdx in {25,161}{\draw[tdbox, fill=white] (\tdx,\tdy) rectangle ++(52,-18);}
    \foreach \tdx in {83,143}{%
      \filldraw[draw=tdNET, fill=tdNETf, line width=0.4pt] (\tdx,\tdy-5) rectangle ++(12,-8);}
    \ifnum\tdj=#2\else
      \draw[draw=tdMUTE!60, line width=0.45pt, dash pattern=on 1.5pt off 1.2pt] (95,\tdy-9) -- (143,\tdy-9);
    \fi
    \node[font=\tdsmall, anchor=west, inner sep=0pt] at (27.5,\tdy-4.4) {GPU \tdj};
    \node[font=\tdsmall, anchor=west, inner sep=0pt] at (163.5,\tdy-4.4) {GPU \tdj};}
  \end{scope}}
\def\tdRail#1{%
  \tdTop{#1}%
  \draw[draw=tdNET, line width=0.7pt] (77,\tdy-9) -- (83,\tdy-9);
  \draw[tdnet] (95,\tdy-9) -- (143,\tdy-9);
  \draw[draw=tdNET, line width=0.7pt, ->] (155,\tdy-9) -- (160.5,\tdy-9);}
\def\tdStore#1#2{
  \foreach \tdk in {0,...,15}{%
    \ifnum\tdk=1
      \filldraw[tdtile, fill=tdID0] (#1+2.6*\tdk,#2-9.5) rectangle ++(2.6,-5.5);
    \else
      \filldraw[draw=tdRULE, line width=0.25pt, fill=white] (#1+2.6*\tdk,#2-9.5) rectangle ++(2.6,-5.5);
    \fi}}

\node[font=\tdbold, anchor=base west, inner sep=0pt] at (0,-6) {(a) All-gather};
\tdScaffold{0}{1}
\node[font=\tdtiny, text=tdNET, anchor=south, inner sep=1pt] at (89,-33) {NIC};
\foreach \tdj in {0,...,3}{\tdTop{\tdj}\tdStore{29}{\tdy}\tdStore{165}{\tdy}}
\tdTop{1}
\draw[tdce, ->] (25,\tdy-6) -- (11.5,\tdy-6);
\foreach \tdj in {0,...,3}{\tdTop{\tdj}\draw[tdnvs, ->] (11,\tdy-13) -- (24.5,\tdy-13);}
\tdRail{1}
\tdTop{1}
\draw[tdce, ->] (213,\tdy-6) -- (226.5,\tdy-6);
\foreach \tdj in {0,...,3}{\tdTop{\tdj}\draw[tdnvs, ->] (227,\tdy-13) -- (213.5,\tdy-13);}
\tdTop{1}
\node[tdstep] at (18,\tdy+0.2) {1};
\node[tdstep] at (119,\tdy-2.5) {2};
\node[tdstep] at (220,\tdy+0.2) {3};
\node[font=\tdtiny, text=tdNET, fill=white, anchor=north, inner sep=1pt] at (119,\tdy-14) {RDMA, $\lceil\log_2 M\rceil$ rounds};

\begin{scope}[shift={(260,0)}]
\node[font=\tdbold, anchor=base west, inner sep=0pt] at (0,-6) {(b) Reduce-scatter};
\end{scope}
\tdScaffold{260}{2}
\begin{scope}[shift={(260,0)}]
  \foreach \tdj in {0,1,3}{%
    \tdTop{\tdj}%
    \filldraw[tdtile, fill=tdCMPf] (66,\tdy-9.5) rectangle ++(7,-5.5);
    \draw[tdce, ->] (25,\tdy-6) -- (11.5,\tdy-6);}
  \tdTop{0}
  \node[font=\tdtiny, text=tdMUTE, anchor=south, inner sep=0.8pt] at (69.5,\tdy-9.5) {bf16};
  \tdTop{2}
  \draw[tdce, ->] (11,\tdy-6) -- (24.5,\tdy-6);
  \filldraw[tdbox, draw=tdSMd, fill=tdSM!70] (50,\tdy-8.5) rectangle ++(24,-8);
  \node[font=\tdtiny, inner sep=0pt] at (62,\tdy-12.5) {TMA $\Sigma$};
  \node[tdstep] at (18,\tdy+3.5) {4};
  \node[tdstep] at (43,\tdy-12.5) {5};
  \tdRail{2}
  \tdTop{2}
  \node[tdstep] at (119,\tdy-2.5) {6};
  \node[font=\tdtiny, text=tdNET, fill=white, anchor=north, inner sep=1pt] at (119,\tdy-14) {fp32 partial};
  \filldraw[tdtile, fill=tdID0] (199,\tdy-9.5) rectangle ++(10,-5.5);
  \tdTop{3}
  \filldraw[tdtile, fill=tdID1] (199,\tdy-9.5) rectangle ++(10,-5.5);
  \foreach \tdj in {2,3}{\tdTop{\tdj}\draw[tdnvs, ->] (213,\tdy-6) -- (226.5,\tdy-6);}
  \foreach \tdj in {0,...,3}{%
    \tdTop{\tdj}%
    \draw[tdnvs, ->] (227,\tdy-13) -- (213.5,\tdy-13);
    \filldraw[draw=tdSMd, fill=tdSM!70, line width=0.35pt] (185,\tdy-9.5) rectangle ++(10,-5.5);
    \node[font=\tdtiny, inner sep=0pt] at (190,\tdy-12.25) {$\Sigma$};}
  \tdTop{2}
  \node[tdstep] at (220,\tdy+3.5) {7};
\end{scope}

\begin{scope}[shift={(0,-136)}, font=\tdsmall]
  \draw[tdce, ->] (0,0) -- (12,0); \node[anchor=west, inner sep=1.5pt] at (12,0) {copy-engine write};
  \draw[tdnvs, ->] (82,0) -- (94,0); \node[anchor=west, inner sep=1.5pt] at (94,0) {NVSwitch multicast or $\Sigma$};
  \draw[tdnet, ->] (196,0) -- (208,0); \node[anchor=west, inner sep=1.5pt] at (208,0) {RDMA on one rail};
  \filldraw[draw=tdSMd, fill=tdSM!70, line width=0.35pt] (282,-2.75) rectangle ++(8,5.5);
  \node[anchor=west, inner sep=1.5pt] at (290,0) {SM work, few CTAs};
  \filldraw[tdtile, fill=tdID0] (366,-2.75) rectangle ++(8,5.5);
  \node[anchor=west, inner sep=1.5pt] at (374,0) {data from node 0};
  \filldraw[tdtile, fill=tdID1] (436,-2.75) rectangle ++(8,5.5);
  \node[anchor=west, inner sep=1.5pt] at (444,0) {from node 1};
\end{scope}
\end{tikzpicture}
  \caption{\system's collective paths (two nodes; four of eight GPUs
  shown). (a) All-gather: local multicast~\tdref{1}, rail
  RDMA~\tdref{2}, destination multicast~\tdref{3}. (b) Reduce-scatter:
  copy-engine staging~\tdref{4}, TMA \dtype{fp32} sums~\tdref{5}, rail
  RDMA~\tdref{6}, NVSwitch reduction~\tdref{7}.}
  \Description{Two panels, each with two nodes of four GPUs, an NVSwitch
  per node, and rail NICs, showing numbered all-gather and reduce-scatter
  steps.}
  \label{fig:arch-overview}
\end{figure*}

\section{Motivation and Design Principles}
\label{sec:principles}

\Cref{tab:notation} summarizes the notation used throughout the paper.

\begin{table}[t]
  \centering
  \caption{Notation.}
  \label{tab:notation}
  \small
  \setlength{\tabcolsep}{4pt}
  \begin{tabular*}{\columnwidth}{@{\extracolsep{\fill}}ll@{}}
    \toprule
    \textbf{Symbol} & \textbf{Meaning} \\
    \midrule
    $M$, $L$, $N=ML$ & nodes, GPUs per NVSwitch domain, ranks \\
    $r=nL+l$ & node-major rank of local GPU $l$ on node $n$ \\
    $P$, $p$ & parameters in an FSDP2 group, parameter index \\
    $R_p$, $K_p$ & padded per-rank rows and row length of $p$ \\
    $s_p$, $o_p$ & per-rank shard bytes and prefix offset of $p$ \\
    $S$, $NS$ & per-rank packed shard bytes; unsharded group bytes \\
    $S_{\mathrm{fp32}}$ & per-rank \dtype{fp32} reduce-scatter output, $2S$ \\
    $B_{\mathrm{pack}}$ & FSDP2's packed \dtype{fp32} reduce-scatter input \\
    $k$ & maximum concurrent rail transfers per round \\
    $W_{\mathrm{wave}}$, $W$ & contributors per wave, number of waves \\
    $n_{\mathrm{out}}$ & outbox blocks per bank \\
    $n_{\mathrm{stg}}$ & local \dtype{bf16} staging slots \\
    $n_{\mathrm{par}}$ & inbox buffer generations (1 or 2) \\
    $n_{\mathrm{part}}$ & sub-slice passes per reduce-scatter \\
    $n_{\mathrm{ctx}}$ & RMA contexts per communicator \\
    \bottomrule
  \end{tabular*}
\end{table}

\subsection{The Layout Mismatch in FSDP2}
\label{sec:anatomy}

Consider a cluster of $M$ nodes with $L$ GPUs per NVSwitch domain and
$N = ML$ total ranks. We use node-major rank $r = nL + l$ for
$n \in [0, M)$ and $l \in [0, L)$. An FSDP2 parameter group
(\code{FSDPParamGroup}) manages $P$ parameters sharded along dimension~0.
Logically reshape parameter $p \in [0, P)$ as $(D_p, K_p)$, where $D_p$
is its dimension-0 extent and $K_p$ is the product of its trailing
dimensions. Its padded per-rank row count
is $R_p = \lceil D_p / N \rceil$, shard byte size is
$s_p = R_p K_p b_{\mathrm{ag}}$ (where $b_{\mathrm{ag}} \in \{2, 4\}$ is
the all-gather item size), prefix byte offset is
$o_p = \sum_{i < p} s_i$, and total packed shard size is
$S = \sum_{p=0}^{P-1} s_p$.

\paragraph{The forward/recompute all-gather boundary.}
Stock FSDP2's \code{foreach\_\allowbreak all\_\allowbreak gather} packs
(and, under mixed precision, casts to \dtype{bf16}) the rank's $P$
shards into an $S$-byte send buffer. The rank-major all-gather output
places rank $r$'s shard of parameter $p$ at byte offset $rS + o_p$, but
compute kernels need each unsharded parameter as a contiguous
$(D_p, K_p)$ view, so that slice must land at the
\emph{parameter-major landing offset}
\begin{equation}
  \mathrm{Landing}(p, r) \;=\; o_p N + r s_p.
  \label{eq:landing}
\end{equation}
FSDP2 bridges the two with
\code{split\_\allowbreak with\_\allowbreak sizes\_\allowbreak copy} on
the compute stream right before the layer runs: $2NS$ bytes of HBM
traffic on the critical path, with both $NS$-byte buffers live.

\paragraph{The backward reduce-scatter boundary.}
After a layer's backward pass, autograd leaves $P$ unsharded gradients
$G_p^{(u)}$ (typically \dtype{bf16}) on each source rank $u \in [0, N)$.
Let $b_{\mathrm{grad}} \in \{2, 4\}$ be the gradient item size and
$S_{\mathrm{grad}} = b_{\mathrm{grad}}\sum_{p}R_pK_p$ the packed
per-rank gradient shard size. Destination rank $r$ owns rows
$\mathcal{I}_{p, r} = [r R_p, (r+1) R_p)$ of each $G_p^{(u)}$ (rows past
$D_p$ are zero). Because \code{ncclReduceScatter} reduces one rank-major
buffer in its own datatype and our target configurations reduce in
\dtype{fp32}, stock FSDP2 calls \code{torch.\_chunk\_cat} on the
compute stream to slice,
zero-pad, upcast, and pack all gradients into a transient \dtype{fp32}
buffer of
\begin{equation}
  B_{\mathrm{pack}} \;=\; 4 \sum_{p=0}^{P-1} N R_p K_p
  \;=\; \frac{4}{b_{\mathrm{grad}}} N S_{\mathrm{grad}}
  \label{eq:pack-bytes}
\end{equation}
bytes ($2NS_{\mathrm{grad}}$ for \dtype{bf16} gradients) before the
reduce-scatter can begin. Both copy kernels also wait on the compute
stream for a host-to-device (H2D) copy of their address and size tables.
This metadata copy is serialized with activation-prefetch H2D transfers
and can therefore wait behind a much larger transfer, resulting in
head-of-line blocking.

The two boundaries are inverse layout transformations that leave the
mathematical collective unchanged. \system removes both by making
placement part of the collective address map.

\subsection{SM Wave Quantization}
\label{sec:wave-quant}

FSDP inserts collectives around every layer and overlaps them with
neighboring compute on side streams, which provide concurrency but not
isolation: SM-based communication kernels compete with model kernels
for SMs, registers, and shared memory. High-throughput CUTLASS/CuTe
tensor-memory-accelerator (TMA) GEMMs use over 200~KB of shared memory
per thread block (CTA), so at most one CTA fits
on each of the $N_{\mathrm{SM}} = 148$ SMs.

These GEMMs also launch CTAs as thread-block clusters, which the
hardware co-schedules within one GPU processing cluster
(GPC)~\cite{nvidia2025cudaguide}. A GPU therefore offers only
$n_{\mathrm{slot}}$ slots for clusters of a given size and footprint.
With $n_{\mathrm{gemm}}$ GEMM clusters and $n_{\mathrm{blk}}$ slots
blocked by resident communication CTAs, a first-order model predicts
\begin{equation}
  T_{\mathrm{gemm}} \;\approx\;
  \left\lceil \frac{n_{\mathrm{gemm}}}{n_{\mathrm{slot}} - n_{\mathrm{blk}}} \right\rceil T_{\mathrm{wave}},
  \label{eq:wave-quant}
\end{equation}
where $T_{\mathrm{wave}}$ is the duration of one isolated
wave~\cite{nvidia2023matmul}. A clustered communication kernel blocks
at least one slot per cluster; single CTAs block none when they land
outside the GEMM's slots and up to one each when they spread across
GPCs.

Consider a small-$M$ GEMM ($256 \times 8192 \times 8192$), as in MoE
experts or narrow projections at small token counts. Its 128 CTAs form
$n_{\mathrm{gemm}} = 32$ four-CTA clusters, and each CTA's 213~KB of
shared memory and 255 registers per thread fill an SM. The GPU holds
$n_{\mathrm{slot}} = 33$ such clusters, so the GEMM runs in one wave
with one spare slot and 16 SMs outside any slot
(\Cref{fig:wave-quant}a). NCCL~2.31 launches four-CTA clusters by
default; its 32~MiB reduce-scatter on 8 nodes runs 16 CTAs, whose four
clusters take the spare slot and three GEMM slots. The 12 displaced GEMM
CTAs run as a second wave (\Cref{fig:wave-quant}b), slowing the GEMM by
$69\%$. Capping NCCL at 8 CTAs still blocks two slots ($67.8\%$).
With cluster size 1, the eight CTAs fit outside the GEMM's slots
($6.3\%$; \Cref{fig:occupancy}), but in the same runs the reduce-scatter
takes 13.0~ms, about twice as long as \system's.

We verify this accounting with an SM-id probe: a kernel with the
GEMM's cluster shape and per-CTA footprint whose CTAs record their SM
id and entry and exit timestamps. The measured occupancy is 33
clusters. During a 32~MiB reduce-scatter, default NCCL's four
clusters push exactly the predicted 12 probe CTAs into a second wave
in 64\% of the probe windows (Stock FSDP2: 65\%); capping NCCL at
8~CTAs pushes 4, and a 4-CTA cap and both \system configurations push
none (\Cref{tab:smid}).

\begin{figure}[t]
  \centering
\begin{tikzpicture}[td, x=1pt, y=1pt]
\def\tdc{3.5}\def\tdp{8}\def\tdTop{-17}\def\tdRow{-49.2}
\def\tdCell#1#2#3{%
  \coordinate (tdq) at (#1+0.3,#2-0.3);
  \ifcase#3
    \draw[draw=tdRULE, line width=0.3pt] (tdq) rectangle ++(\tdc-0.6,-\tdc+0.6);
  \or
    \fill[tdCMP] (tdq) rectangle ++(\tdc-0.6,-\tdc+0.6);
  \or
    \fill[tdSM] (tdq) rectangle ++(\tdc-0.6,-\tdc+0.6);
  \or
    \filldraw[draw=tdSMd, line width=0.35pt, dash pattern=on 0.6pt off 0.45pt, fill=tdSMf]
      (tdq) rectangle ++(\tdc-0.6,-\tdc+0.6);
  \fi}
\def\tdSlot#1#2#3{%
  \tdCell{#1}{#2}{#3}\tdCell{#1+\tdc}{#2}{#3}%
  \tdCell{#1}{#2-\tdc}{#3}\tdCell{#1+\tdc}{#2-\tdc}{#3}}
\def\tdGrid#1#2#3{%
  \foreach \tdk in {0,...,31}{%
    \pgfmathtruncatemacro\tdrow{\tdk/8}%
    \pgfmathtruncatemacro\tdcol{\tdk-8*\tdrow}%
    \pgfmathtruncatemacro\tdst{(\tdk>=33-#2)?2:1}%
    \tdSlot{#1+\tdcol*\tdp}{\tdTop-\tdrow*\tdp}{\tdst}}%
  \pgfmathtruncatemacro\tdsp{(#2>0)?2:((#3>0)?3:0)}%
  \tdSlot{#1}{\tdRow}{\tdsp}%
  \pgfmathtruncatemacro\tdlo{#3*3}%
  \foreach \tdk in {0,...,15}{%
    \pgfmathtruncatemacro\tdrow{\tdk/8}%
    \pgfmathtruncatemacro\tdcol{\tdk-8*\tdrow}%
    \tdCell{#1+\tdp+2+\tdcol*\tdc}{\tdRow-\tdrow*\tdc}{\tdlo}}}
\def\tdTitle#1#2#3{%
  \node[font=\tdbold, anchor=base west, inner sep=0pt] at (#1,-6) {#2};
  \node[font=\tdsmall, text=tdMUTE, anchor=base west, inner sep=0pt] at (#1,-13) {#3};}
\def\tdBar#1#2{
  \filldraw[draw=tdINK!60, line width=0.3pt, fill=tdCMP] (#1,-71) rectangle ++(26,-6);
  \ifnum#2=2
    \filldraw[draw=tdINK!60, line width=0.3pt, fill=tdCMP] (#1+26,-71) rectangle ++(26,-6);
    \fill[pattern=north east lines, pattern color=tdINK!80] (#1+26,-71) rectangle ++(26,-6);
  \fi}

\tdTitle{0}{(a) GEMM alone}{32 clusters of 4 CTAs}
\tdGrid{0}{0}{0}
\tdTitle{86}{(b) + NCCL RS}{16 CTAs in 4 clusters}
\tdGrid{86}{4}{0}
\tdTitle{172}{(c) + TIDE RS}{16 + 4 single CTAs}
\tdGrid{172}{0}{1}

\foreach \tdk in {0,1,2}{%
  \tdSlot{86+\tdk*\tdp}{-59.5}{1}
  \fill[pattern=north east lines, pattern color=tdINK!80] (86.3+\tdk*\tdp,-59.8) rectangle ++(2*\tdc-0.6,-2*\tdc+0.6);}
\node[font=\tdsmall, text=tdMUTE, anchor=west, inner sep=1pt] at (86+3*\tdp,-63) {spill: 12 CTAs};

\tdBar{0}{1}\tdBar{86}{2}\tdBar{172}{1}
\node[font=\tdsmall, anchor=west, inner sep=1.5pt] at (26,-74) {1 wave};
\node[font=\tdsmall, anchor=west, inner sep=1.5pt, text=tdTAX] at (138,-74) {2 waves};
\node[font=\tdsmall, anchor=west, inner sep=1.5pt] at (198,-74) {1 wave};
\node[font=\tdtiny, text=tdMUTE, anchor=north west, inner sep=0pt] at (0,-79.5) {GEMM latency};

\begin{scope}[shift={(0,-92)}, font=\tdsmall]
  \fill[tdCMP] (0,0) rectangle ++(\tdc,\tdc);
  \node[anchor=west, inner sep=1.5pt] at (\tdc,0.5*\tdc) {GEMM CTA};
  \fill[tdSM] (46,0) rectangle ++(\tdc,\tdc);
  \node[anchor=west, inner sep=1.5pt] at (46+\tdc,0.5*\tdc) {resident comm.\ CTA};
  \filldraw[draw=tdSMd, line width=0.35pt, dash pattern=on 0.6pt off 0.45pt, fill=tdSMf]
    (122,0) rectangle ++(\tdc,\tdc);
  \node[anchor=west, inner sep=1.5pt] at (122+\tdc,0.5*\tdc) {transient CTA};
  \draw[draw=tdRULE, line width=0.3pt] (184,0) rectangle ++(\tdc,\tdc);
  \node[anchor=west, inner sep=1.5pt] at (184+\tdc,0.5*\tdc) {idle SM};
\end{scope}
\end{tikzpicture}
  \caption{SM wave quantization for the $256\times8192\times8192$ GEMM
  (schematic; slot counts and CTA displacement measured with an SM-id
  probe, \Cref{tab:smid}). Each $2\times2$ block is a slot for one four-CTA
  cluster; the bottom row holds the spare slot and the 16 SMs outside
  any slot.}
  \Description{Three panels, each with 33 two-by-two cluster slots and
  16 single SMs, showing GEMM CTAs, resident communication clusters,
  and transient TIDE CTAs, with GEMM latency bars.}
  \label{fig:wave-quant}
\end{figure}
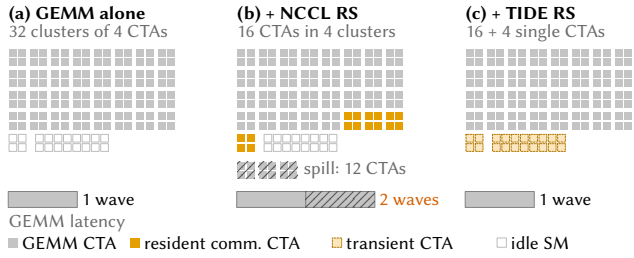

Placement determines which kernels displace a wave; residence time
determines how often. NCCL's ring and symmetric-memory reduce-scatter
kernels stay resident for the whole transfer, polling peer progress, so
every overlapping GEMM pays. Placement also separates the two NCCL
kernels: the symmetric-memory kernel is resident just as long---its
CTAs appear in over 90\% of probe windows---yet displaces almost no
GEMM CTA, because its 13 CTAs land on the SMs outside the GEMM's
cluster slots. Its 11\% one-wave slowdown is therefore not wave
quantization but the HBM direction cost of switch-reduction reads
(\Cref{sec:direction-tax}). A spin kernel that holds its CTAs for the
whole GEMM shows the cost of residence: with its CTAs spread across
GPCs, 8 single CTAs slow the GEMM by $66\%$ and 20 by $271\%$
(\Cref{fig:occupancy}). \system's 16 sum and 4 landing CTAs each live
0.2--0.4~ms and are resident for about a third of a reduce-scatter;
together they slow the GEMM by less than $5\%$. A grouped GEMM that fills
every SM (Qwen3-30B-A3B experts, 128 groups of 512 rows) has no spare
capacity: every NCCL and symmetric-memory reduce-scatter we measured
slows it by about $70\%$, even with cluster size 1 or a 4-CTA cap,
whereas \system's training configuration slows it by about $5\%$.

\begin{figure}[t]
  \centering
  \includegraphics[width=\columnwidth]{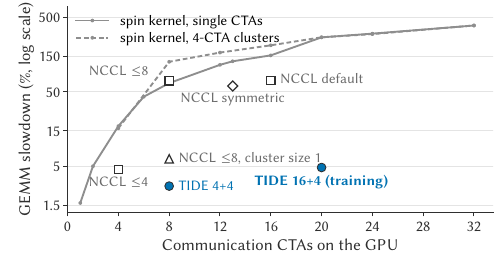}
  \caption{Slowdown of the $256\times8192\times8192$ GEMM versus
  communication CTAs at a fixed SM clock. Lines: spin kernels resident
  for the whole GEMM. Markers: 32~MiB reduce-scatters on 8 nodes (64-GPU
  median), with \system's short-lived CTAs at their peak count (labels:
  sum + landing CTAs; 16+4 is the training configuration); runs
  separate from \Cref{tab:eval-overlap}.}
  \Description{Log-scale slowdown versus CTA count. Two gray curves rise
  from 2 to 390 percent; NCCL markers at 8 to 16 CTAs sit near 60 to 71
  percent, while NCCL with 4 CTAs or cluster size 1 and both TIDE
  configurations stay below 7 percent.}
  \label{fig:occupancy}
\end{figure}

\Cref{tab:eval-overlap} and \Cref{fig:eval-cta} show the resulting
trade-off. Default NCCL
finishes the reduce-scatter in 6.75~ms but slows the GEMM by $69.0\%$;
capping it at 4 CTAs, one cluster in the spare slot, cuts the slowdown
to $2.7\%$ but stretches the reduce-scatter to 26.35~ms. \system's
training configuration matches that $2.9\%$ while finishing in
6.96~ms, faster than default NCCL: rather than capping an SM
collective, it moves transport and peer waiting off SMs and spends its
CTA budget only on local \dtype{fp32} accumulation and landing.

Stream priority cannot clear the way either: CUDA does not preempt
resident CTAs, so a high-priority kernel waits for a persistent GEMM to
retire its CTAs. GEMMs built on Cluster Launch Control~\cite{ptx9} admit
it within two tile boundaries, which suits short-lived CTAs
(\Cref{app:clc}).

\subsection{Limits of Symmetric Input Windows}
\label{sec:why-not-symm-input}

NVSwitch reduces in the switch via PTX
\code{multimem.\allowbreak ld\_reduce}~\cite{ptx9}, so one might place
unsharded gradients in a symmetric window and let each owner reduce the
$L=8$ local copies with a 2-CTA kernel. With per-parameter \dtype{bf16}
gradients and \dtype{fp32} reduction, this fails either way:
\begin{enumerate}
  \item \textbf{A \dtype{bf16} window breaks \dtype{fp32} semantics.}
  The \code{acc::f32.\allowbreak bf16x2} form of
  \code{multimem.\allowbreak ld\_reduce} accumulates in \dtype{fp32} but
  returns \dtype{bf16x2}, so every node partial is
  rounded to \dtype{bf16} before cross-node accumulation, and small
  contributions can be lost as $M$ grows.
  \item \textbf{An \dtype{fp32} window is prohibitive.}
  \code{multimem.\allowbreak ld\_reduce.\allowbreak add.v4.f32} needs
  \dtype{fp32} inputs: a persistent $2NS$-byte symmetric allocation per
  active group ($2\times$ the unsharded \dtype{bf16} group), plus either
  rewriting every gradient-producing kernel to write \dtype{fp32} into
  window offsets or an SM cast-pack pass that reintroduces the traffic,
  SMs, and buffers \system removes.
\end{enumerate}
\system instead consumes FSDP2's native \dtype{bf16} gradients, stages
bounded chunks, and accumulates once in \dtype{fp32} locally.

\subsection{The HBM Direction Tax}
\label{sec:direction-tax}

Even with an \dtype{fp32} window, traffic direction matters: per byte
served, remote service reads disrupt compute $2.2$--$4\times$ more
than inbound copy-engine writes, depending on the clock regime. In a $2 \times 2$ matrix of
initiator (SM or copy engine) and direction (pull or push), GPUs 1--7
generate traffic against GPU~0 while it runs a 256~MiB \dtype{bf16}
memory-bound victim (\Cref{tab:hbm-direction}).

\begin{table}[t]
  \centering
  \caption{HBM interference by traffic direction.}
  \label{tab:hbm-direction}
  \small
  \begin{tabular*}{\columnwidth}{@{\extracolsep{\fill}}lcccc@{}}
    \toprule
    & \multicolumn{2}{c}{\textbf{Served traffic (GB/s)}}
    & \multicolumn{2}{c}{\textbf{Victim slowdown}} \\
    \cmidrule(lr){2-3}\cmidrule(l){4-5}
    \textbf{Initiator} & \textbf{Pull} & \textbf{Push}
    & \textbf{Pull} & \textbf{Push} \\
    \midrule
    SM & 788 & 721 & $46.7\%$ & $12.1\%$ \\
    Copy engine & 838 & 798 & $43.9\%$ & $11.0\%$ \\
    \bottomrule
  \end{tabular*}\par\vspace{3pt}
  {\footnotesize\raggedright GPUs 1--7 of one NVSwitch node drive traffic
  at GPU~0, which runs a 256~MiB memory-bound victim. Slowdown: median
  victim kernel time during 10~s of traffic relative to the pre- and
  post-traffic baselines.\par}
\end{table}

Normalized by served bandwidth, pulls slow the victim by
$5.2\text{--}5.9\%$ per 100~GB/s and pushes by $1.4\text{--}1.7\%$. The
asymmetry is not a power-cap artifact: at a fixed 1200~MHz SM clock,
pulls cost $4.46$ (SM) and $3.81$ (copy engine) versus $2.04$ and
$1.21\%$ per 100~GB/s for pushes, a $2.2\text{--}3.2\times$ per-byte
gap. Across the two clock regimes, a pull therefore costs
$2.2$--$4\times$ as much per byte as a push, and we quote this range
rather than a single factor. A real 2-CTA switch-reduction collective slows the same victim by $27\%$,
\system-RS by $2.6\%$. L2-resident victims (32--64~MiB) slow by at most
$3.4\%$, so the cost is HBM contention, not L2 eviction. DRAM counters
show no extra bytes, consistent with interleaved service reads lowering
DRAM efficiency while inbound writes are coalesced before writeback.

\paragraph{Architectural synthesis.}
These findings shape \system-RS's two levels
(\Cref{fig:arch-overview}b). Within a node, copy engines push unreduced
\dtype{bf16} gradients~\tdref{4}---half the bytes, in the
low-interference direction---into small staging chunks, and 4--16 TMA
CTAs~\tdref{5} accumulate the $8 \to 1$ sum in \dtype{fp32}. Across
nodes, each node contributes one 4-byte \dtype{fp32} partial per element
instead of eight 2-byte values ($4\times$ less input); partials arrive
over rail RDMA~\tdref{6}, and
\code{multimem.\allowbreak ld\_reduce.\allowbreak add.v4.f32}~\tdref{7}
lets 1--4 landing CTAs fold up to 8 nodes per wave in \dtype{fp32}.
Switch reads thus see only the $4\times$ smaller partial stream.

\subsection{Design Requirements and Scope}
\label{sec:goals}

\system targets six requirements:
\begin{itemize}
  \item \textbf{G1 (Preserve the FSDP2 contract):} per-parameter
  dimension-0 \code{DTensor} shards, \dtype{fp32} reduction, and
  unmodified compute kernels.
  \item \textbf{G2 (Placement, not packing):} land all-gather shards in
  parameter-major storage and consume per-parameter gradients directly,
  keeping only required dtype casts.
  \item \textbf{G3 (Bound SM interference):} no SMs for all-gather, and a
  reduce-scatter CTA budget independent of node count, with no peer
  polling on SMs in the default path.
  \item \textbf{G4 (Scale across nodes):} $O(\log M)$ all-gather rounds,
  about $O(M/k)$ wave-aligned reduce-scatter rounds with $k$ concurrent
  rail transfers each, and rail communicators over $M$ rather than $N$
  ranks.
  \item \textbf{G5 (Bounded memory):} no persistent full-group
  \dtype{fp32} window, a closed-form memory model, and two time--memory
  modes.
  \item \textbf{G6 (Public NCCL APIs):} public NCCL~2.31+ APIs only;
  unsupported groups fall back to stock collectives.
\end{itemize}

\system targets PyTorch~2.13, dimension-0 sharding, node-major
deployments with up to eight GPUs per NVSwitch domain, rail-aligned RDMA
NICs, NVSwitch \code{multimem}, and compute capability $\ge9.0$ with at
least 207~KB of shared memory per block. It does
not yet support CUDA-graph capture or asynchronous unshard; groups
outside these contracts, including those with more than 256 trainable
parameters, fall back to stock collectives.

\section{TIDE-AG: Layout-Aware All-Gather}
\label{sec:ag}

\system-AG makes parameter-major storage the final collective output:
multi-node shards may transit a bounded receive stage, but FSDP2 never
sees a rank-major buffer or performs copy-out. A copy engine multicasts
each shard through NVSwitch into every local GPU~\tdref{1}, local GPU
$l$ forwards it over rail $l$~\tdref{2}, and a copy engine on each
receiving node multicasts it into place~\tdref{3}
(\Cref{fig:arch-overview}a). The final landing resides in an NCCL
symmetric arena.

\subsection{In-Place Shards and Paced Doubling}
\label{sec:ag-residence}

\paragraph{In-place shard adoption vs.\ mixed-precision cast.}
When a group's sharded dtype matches the all-gather dtype (e.g.,
\dtype{bf16}), \system reserves a persistent $S$-byte resident-shard
region for it in the symmetric arena. At the first forward pass, the
adapter copies each padded shard there once and repoints the sharded
parameter and its \code{DTensor} local view to zero-offset views over
it, so later all-gathers read shards in place without copy-in.
Mixed-precision groups (\dtype{fp32} master shards) keep FSDP2's fused
cast but write its output into the first receive-stage slot once the
previous call releases it.

\paragraph{Single-node multicast and multi-node Bruck recursive doubling.}
On a single node ($M=1$, $N=L\le 8$), the arena is bound to NVSwitch
multicast. After a stream-memory-operation entry fence, rank $r$ issues
one batched copy with $P$ descriptors to each parameter's final
multicast address, so every write reaches all $L$ GPUs once instead of
through $L-1$ peer unicasts. Cumulative completion counters prevent
consumption until every local rank finishes.

Across $M>1$ nodes, FSDP2's default prefetch depth of one uses two
alternating $NS$-byte landings, and a group needs a receive stage of
$(M-1)S$ bytes ($MS$ with a cast); without expert subgroups, both
landings share that stage. Rather than an $(M-1)$-round ring,
\system-AG runs Bruck's recursive doubling~\cite{bruck1997efficient} in
$R = \lceil \log_2 M \rceil$ rounds (\Cref{fig:ag-protocol}a). In round
$d$, node $n$ forwards its first $c_d=\min(2^d, M-2^d)$ accumulated
shard slots to rail peer $(n-2^d)\bmod M$ and receives as many from
$(n+2^d)\bmod M$; since $\sum_d c_d=M-1$, any $M$ is covered. As slots
arrive, a copy-engine stream multicasts them into the parameter-major
landing. Consecutive stage slots come from consecutive rail-peer ranks,
so their rows in parameter $p$ have a constant destination pitch
$L s_p$ and source pitch $S$ (\Cref{fig:ag-protocol}b). \system
therefore merges each run of at least four slots into one strided copy
and splits the last round into at most four groups, cutting descriptors
from $O(MP)$ to $O(P\log M)$.

\paragraph{Lead-based round pacing and per-area credits.}
Unpaced multi-peer puts can trigger RoCE receiver incast and DCQCN rate
cuts. \system therefore starts round $d\ge 1$ only after the round's
target confirms that round $d-1$ has arrived. To hide the roughly
0.12~ms signal round trip, each transfer is split into a head and a
tail (8~MiB of heads in total by default); the receiver signals once
all heads arrive, so the signal overlaps the tail. Without pacing,
training sees $9\times$ more PFC pauses and $2.6\times$ more ECN marks.
Per-area cumulative credits, emitted once every slot of a call has been
scattered, guard receive-stage reuse across calls.

\subsection{Decoupled Stream Completion}
\label{sec:ag-decouple}

Three conditions differ: a landing is free to overwrite, its new
parameters are ready to consume, and a remote receive stage may be
reused. Serializing them on FSDP2's \code{all\_gather\_stream} would
resynchronize local ranks behind the slowest sibling.

\begin{figure}[t]
  \centering
\begin{tikzpicture}[td, x=1pt, y=1pt]
\def\tdSlot#1#2#3#4#5{
  \ifnum#4<0
    \filldraw[draw=tdRULE, line width=0.3pt, fill=white] (#1,#2) rectangle ++(#3,-7);
  \else
    \filldraw[tdtile, fill=tdID#4] (#1,#2) rectangle ++(#3,-7);
    \node[font=\tdtiny, inner sep=0pt] at (#1+0.5*#3,#2-3.5) {#5};
  \fi}

\node[font=\tdbold, anchor=base west, inner sep=0pt] at (-14,-6)
  {(a) Recursive doubling on rail $l$ ($M=5$)};
\fill[tdCMPf, rounded corners=1pt] (41,-15.5) rectangle (77,-60);
\foreach \tdn in {0,...,4}{%
  \node[font=\tdsmall, text=tdMUTE, anchor=base, inner sep=0pt] at (44+36*\tdn+15,-14) {node \tdn};}
\foreach \tdt/\tdk/\tdlab in {0/1/start, 1/2/{round 0 (+1 shard)}, 2/4/{round 1 (+2 shards)},
    3/5/{round 2 (+1 shard)}}{%
  \node[font=\tdsmall, anchor=east, inner sep=1pt] at (39,-21.5-10*\tdt) {\tdlab};
  \foreach \tdn in {0,...,4}{\foreach \tdi in {0,...,4}{%
    \pgfmathtruncatemacro\tdsrc{mod(\tdn+\tdi,5)}%
    \ifnum\tdi<\tdk
      \tdSlot{44+36*\tdn+6*\tdi}{-18-10*\tdt}{6}{\tdsrc}{\tdsrc}
    \else
      \tdSlot{44+36*\tdn+6*\tdi}{-18-10*\tdt}{6}{-1}{}
    \fi}}}

\node[font=\tdbold, anchor=base west, inner sep=0pt] at (-14,-74)
  {(b) One strided multicast copy per parameter ($L=2$)};
\node[font=\tdsmall, anchor=east, inner sep=1pt] at (39,-89.5) {stage};
\node[font=\tdsmall, anchor=east, inner sep=1pt] at (39,-117.5) {$p_0$ storage};
\node[font=\tdtiny, text=tdMUTE, anchor=south, inner sep=0.8pt] at (50,-86) {$p_0$};
\node[font=\tdtiny, text=tdMUTE, anchor=south, inner sep=0.8pt] at (62,-86) {$p_1\,p_2$};
\foreach \tdi in {0,...,4}{%
  \pgfmathsetmacro\tdx{44+24*\tdi}%
  \filldraw[tdtile, fill=tdID\tdi] (\tdx,-86) rectangle ++(12,-7);
  \filldraw[tdtile, fill=tdID\tdi] (\tdx+12,-86) rectangle ++(6,-7);
  \filldraw[tdtile, fill=tdID\tdi] (\tdx+18,-86) rectangle ++(6,-7);
  \node[font=\tdtiny, inner sep=0pt] at (\tdx+6,-89.5) {\tdi};
  \draw[tdce, ->, line width=0.45pt] (\tdx+6,-93.5) -- (44+24*\tdi+18,-113.5);}
\foreach \tdr in {0,...,9}{%
  \pgfmathsetmacro\tdx{44+12*\tdr}%
  \ifodd\tdr
    \pgfmathtruncatemacro\tdsrc{(\tdr-1)/2}%
    \tdSlot{\tdx}{-114}{12}{\tdsrc}{\tdsrc}
  \else
    \fill[pattern=north east lines, pattern color=tdRULE] (\tdx,-114) rectangle ++(12,-7);
    \draw[draw=tdRULE, line width=0.3pt] (\tdx,-114) rectangle ++(12,-7);
  \fi}
\draw[decorate, decoration={brace, mirror, amplitude=2.2pt}, line width=0.4pt, tdMUTE]
  (62,-123) -- (86,-123);
\node[font=\tdtiny, text=tdMUTE, anchor=north, inner sep=1pt] at (74,-125.5) {pitch $L s_0$};
\begin{scope}[font=\tdsmall, text=tdINK]
  \node[anchor=west, inner sep=0pt] at (172,-91) {source pitch $S$};
  \node[anchor=west, inner sep=0pt] at (172,-100) {dest.\ pitch $L s_0$};
  \node[anchor=west, inner sep=0pt, text=tdCE] at (172,-109) {CE multicast to};
  \node[anchor=west, inner sep=0pt, text=tdCE] at (172,-117) {all $L$ local GPUs};
\end{scope}
\end{tikzpicture}
  \caption{\system-AG on one rail. (a) Recursive doubling gathers $M=5$
  shards in $\lceil\log_2 M\rceil=3$ rounds. (b) One strided copy per
  parameter multicasts a run of stage slots into every local GPU; hatched
  rows belong to the other local GPU.}
  \Description{A matrix of stage buffers over recursive-doubling rounds
  for five nodes, and a strided copy from stage slots to parameter
  storage.}
  \label{fig:ag-protocol}
\end{figure}
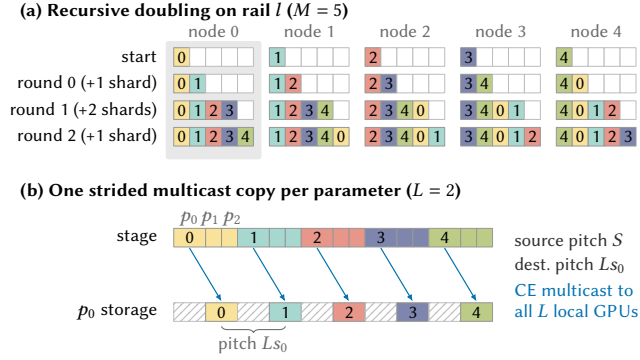

\system-AG therefore separates them across three internal streams. A
readiness stream waits until the landing is free and publishes local
readiness; the caller's stream flushes prior credits, queues the
readiness waits and the $R$ RDMA rounds, and forks scatters to a
copy-engine stream; and a completion stream joins the scatters, runs the
node-wide completion fence, and records the event that FSDP2 waits on
before consuming the parameters. The next call's credit flush waits only
for this rank's previous scatter, not for the previous node-wide fence
or the slowest local sibling.

Safety requires a credit before stage reuse, completion before
consumption, and the same call order on all ranks. Removing any one in
fault injection causes detected corruption or a timeout.

\section{TIDE-RS: Layout-Aware Reduce-Scatter}
\label{sec:rs}

\system-RS applies \system's placement principle to FSDP2 gradient
reduction (\Cref{fig:rs-architecture}). It consumes per-parameter
gradients $\{G_p^{(u)}\}$ without rank-major packing, either directly in
place or from a layout-preserving slot copy (\Cref{sec:two-modes}). It
materializes neither the transient $B_{\mathrm{pack}}$ buffer in
\Cref{eq:pack-bytes} nor a persistent full-group symmetric input window,
and produces on rank $r\in[0,N)$ the concatenated \dtype{fp32} shard
$\mathrm{out}_r = \bigoplus_{p=0}^{P-1} \mathrm{out}_{p, r}$ of total
byte size $S_{\mathrm{fp32}} = 4 \sum_{p=0}^{P-1} R_p K_p$, where
$\bigoplus$ concatenates the per-parameter shards
\begin{equation}
  \mathrm{out}_{p, r} \;\gets\;
  (\mathrm{accumulate} \;?\; \mathrm{out}_{p, r} : 0)
  \;+\; \alpha \sum_{u=0}^{N-1} G_p^{(u)}[\mathcal{I}_{p, r}].
  \label{eq:rs-semantics}
\end{equation}
The output has stock FSDP2's packed per-rank layout; $\alpha$ folds
FSDP2's divide factors and \code{AVG} or \code{PREMUL\_SUM} scaling, and
\code{accumulate} fuses gradient accumulation into the output.

\system-RS has two levels. Level~1 forms one \dtype{fp32} node partial
from the $L$ local contributions---\dtype{bf16} in our target
configuration---using copy-engine pushes and a bounded TMA sum. Level~2
sends these partials over the rails and uses Latin-square placement so
the destination NVSwitch combines up to $L$ nodes per wave.

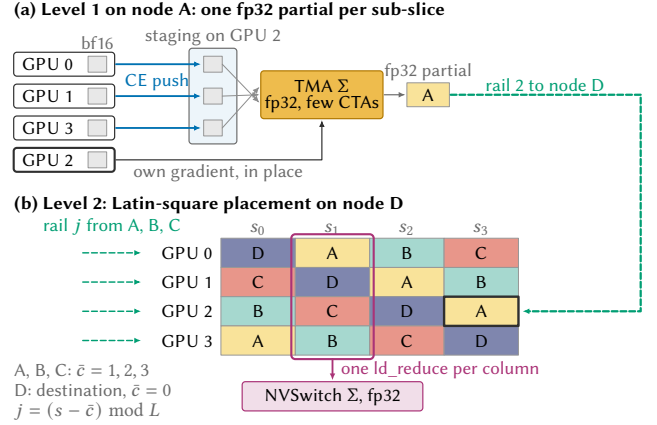
\begin{figure}[t]
  \centering
\begin{tikzpicture}[td, x=1pt, y=1pt]
\node[font=\tdbold, anchor=base west, inner sep=0pt] at (0,-6)
  {(a) Level 1 on node A: one fp32 partial per sub-slice};
\def\tdA{-20}
\foreach \tdj/\tdg in {0/0,1/1,2/3,3/2}{%
  \pgfmathsetmacro\tdy{\tdA-12*\tdj}%
  \ifnum\tdg=2
    \draw[tdbox, line width=0.8pt, fill=white] (0,\tdy) rectangle ++(38,-9);
  \else
    \draw[tdbox, fill=white] (0,\tdy) rectangle ++(38,-9);
  \fi
  \node[font=\tdsmall, anchor=west, inner sep=0pt] at (3,\tdy-4.5) {GPU \tdg};
  \filldraw[tdtile, fill=tdCMPf] (28,\tdy-1.5) rectangle ++(7,-6);}
\node[font=\tdtiny, text=tdMUTE, anchor=south, inner sep=0.8pt] at (31.5,\tdA) {bf16};
\draw[tdbox, draw=tdINK!70, fill=tdCEf!45] (66,\tdA+1.5) rectangle (84,\tdA-34.5);
\node[font=\tdtiny, text=tdMUTE, anchor=south, inner sep=0.8pt] at (75,\tdA+1.5) {staging on GPU 2};
\foreach \tdj in {0,1,2}{%
  \pgfmathsetmacro\tdy{\tdA-12*\tdj}%
  \filldraw[tdtile, fill=tdCMPf] (71.5,\tdy-1.5) rectangle ++(7,-6);
  \draw[tdce, ->] (38.5,\tdy-4.5) -- (70.5,\tdy-4.5);}
\node[font=\tdtiny, text=tdCE, anchor=south, inner sep=0.8pt] at (54,\tdA-15.4) {CE push};
\filldraw[tdbox, draw=tdSMd, fill=tdSM!70] (93,\tdA-7) rectangle (139,\tdA-25);
\node[font=\tdsmall, inner sep=0pt] at (116,\tdA-12.5) {TMA $\Sigma$};
\node[font=\tdtiny, inner sep=0pt] at (116,\tdA-19.5) {fp32, few CTAs};
\foreach \tdj in {0,1,2}{%
  \pgfmathsetmacro\tdy{\tdA-12*\tdj-4.5}%
  \draw[draw=tdINK!55, line width=0.35pt, ->] (78.5,\tdy) -- (92.5,{\tdA-16+(\tdj-1)*3});}
\draw[draw=tdINK, line width=0.45pt, ->] (38.5,\tdA-40.5) -- (116,\tdA-40.5) -- (116,\tdA-25.5);
\node[font=\tdtiny, text=tdMUTE, anchor=north, inner sep=0.8pt] at (77,\tdA-41) {own gradient, in place};
\filldraw[tdtile, fill=tdID0] (148,\tdA-11.5) rectangle ++(16,-9);
\node[font=\tdtiny, inner sep=0pt] at (156,\tdA-16) {A};
\node[font=\tdtiny, text=tdMUTE, anchor=south, inner sep=0.8pt] at (156,\tdA-11.5) {fp32 partial};
\draw[draw=tdINK!70, line width=0.45pt, ->] (139.5,\tdA-16) -- (147.5,\tdA-16);
\node[font=\tdtiny, text=tdNET, anchor=south, inner sep=1pt] at (200,\tdA-16) {rail 2 to node D};

\node[font=\tdbold, anchor=base west, inner sep=0pt] at (0,-79)
  {(b) Level 2: Latin-square placement on node D};
\def\tdGx{78}\def\tdGy{-90}\def\tdCw{28}\def\tdCh{11}
\foreach \tds in {0,...,3}{%
  \node[font=\tdsmall, text=tdMUTE, anchor=south, inner sep=1pt]
    at (\tdGx+\tdCw*\tds+0.5*\tdCw,\tdGy) {$s_\tds$};}
\foreach \tdj in {0,...,3}{%
  \pgfmathsetmacro\tdy{\tdGy-\tdCh*\tdj}%
  \node[font=\tdsmall, anchor=east, inner sep=1pt] at (\tdGx-2,\tdy-0.5*\tdCh) {GPU \tdj};
  \draw[tdnet, ->, line width=0.55pt] (26,\tdy-0.5*\tdCh) -- (48,\tdy-0.5*\tdCh);
  \foreach \tds in {0,...,3}{%
    \pgfmathtruncatemacro\tdc{mod(\tds-\tdj+4,4)}%
    \pgfmathtruncatemacro\tdid{ifthenelse(\tdc==0,3,\tdc-1)}%
    \filldraw[tdtile, fill=tdID\tdid] (\tdGx+\tdCw*\tds,\tdy) rectangle ++(\tdCw,-\tdCh);
    \node[font=\tdsmall, inner sep=0pt] at (\tdGx+\tdCw*\tds+0.5*\tdCw,\tdy-0.5*\tdCh)
      {\ifcase\tdid A\or B\or C\or D\fi};}}
\node[font=\tdtiny, text=tdNET, anchor=south, inner sep=1pt] at (37,\tdGy) {rail $j$ from A, B, C};
\draw[draw=tdINK, line width=0.9pt] (\tdGx+3*\tdCw,\tdGy-2*\tdCh) rectangle ++(\tdCw,-\tdCh);
\draw[tdnet, ->] (164.5,\tdA-16) -- (236,\tdA-16) -- (236,\tdGy-2.5*\tdCh) -- (\tdGx+4*\tdCw+1.5,\tdGy-2.5*\tdCh);
\draw[tdnvs, line width=0.8pt, rounded corners=1.5pt]
  (\tdGx+\tdCw-1.5,\tdGy+1.5) rectangle (\tdGx+2*\tdCw+1.5,\tdGy-4*\tdCh-1.5);
\draw[tdnvs, ->] (\tdGx+1.5*\tdCw,\tdGy-4*\tdCh-1.5) -- (\tdGx+1.5*\tdCw,\tdGy-4*\tdCh-9.5);
\filldraw[tdbox, draw=tdNVS, fill=tdNVSf] (\tdGx+1.5*\tdCw-34,\tdGy-4*\tdCh-10) rectangle ++(68,-11);
\node[font=\tdsmall, inner sep=0pt] at (\tdGx+1.5*\tdCw,\tdGy-4*\tdCh-15.5) {NVSwitch $\Sigma$, fp32};
\node[font=\tdtiny, text=tdNVS, anchor=west, inner sep=1pt] at (\tdGx+1.5*\tdCw+2,\tdGy-4*\tdCh-5.5)
  {one ld\_reduce per column};
\node[font=\tdtiny, text=tdMUTE, anchor=west, inner sep=0pt] at (0,\tdGy-4*\tdCh-6) {A, B, C: $\bar c=1,2,3$};
\node[font=\tdtiny, text=tdMUTE, anchor=west, inner sep=0pt] at (0,\tdGy-4*\tdCh-13.5) {D: destination, $\bar c=0$};
\node[font=\tdtiny, text=tdMUTE, anchor=west, inner sep=0pt] at (0,\tdGy-4*\tdCh-21) {$j=(s-\bar c)\bmod L$};
\end{tikzpicture}
  \caption{Two-level \system-RS. (a) Copy engines stage local
  \dtype{bf16} slices and bounded TMA CTAs form \dtype{fp32} node
  partials. (b) Latin-square placement distributes contributors across
  destination GPUs so NVSwitch reduces each aligned sub-slice.}
  \Description{Level-1 node-local reduction into an fp32 partial, and a
  four-by-four Latin square of contributor nodes on the destination
  node's GPUs.}
  \label{fig:rs-architecture}
\end{figure}

\subsection{Level 1: Node-Local Reduction}
\label{sec:rs-level1-push}
\label{sec:rs-tma-kernel}

For each destination node, every local GPU owns one sub-slice spanning
its $L$ destination ranks. The other $L-1$ local ranks push their
\dtype{bf16} contributions into the owner's bounded staging area, while
the owner reads its own contribution directly from the per-parameter
gradient (\Cref{fig:rs-architecture}a). Staging is processed and reused in chunks; device-side counters
prevent a chunk from being overwritten before the owner finishes reading
it.

\paragraph{Coalescing the destination-rank axis.}
Flattening every destination row separately requires $(L-1)PML$ copy
descriptors per rank per call ($107{,}520$ at $(L,M,P)=(8,64,30)$). The
$L$ destination ranks on one node are consecutive, however, so their
slices within each parameter and their staging rows have constant source
and destination pitches. One strided descriptor therefore replaces $L$
flat copies for each owner, parameter, and destination node (plus a
short tail for a final padded row), cutting the count to $(L-1)PM$. On
six nodes, host issue time drops from 2.0 to 0.8~ms for $P=16$ and from
3.6 to 1.2~ms for $P=30$.

\paragraph{Keeping pace with the rail.}
At $L=8$, each 4-byte \dtype{fp32} node partial requires eight
\dtype{bf16} reads (16 bytes) and one \dtype{fp32} write (4 bytes), so
Level~1 must process $5\times$ the NIC output rate: 237.5~GB/s of
logical read-plus-write traffic for a 47.5~GB/s rail. A warp-specialized
TMA pipeline raises the kernel-only rate from 32 to 99~GB/s per CTA in
this metric; on six nodes, a reduce-scatter takes 4.81~ms with four sum
CTAs versus 4.60~ms with eight, so four CTAs keep up with a 400~Gb/s
rail. The sum CTAs retire when the read phase completes and never poll
peers, so extra CTAs add transient contention without reserving SMs for
the network transfer; the pipeline can use more of them to meet an
earlier gradient-read deadline.

\subsection{Level 2: Latin-Square NVSwitch Waves}
\label{sec:rs-level2}

SM-resident baselines keep 13--16 CTAs active for the whole transfer
(\Cref{sec:wave-quant}). \system-RS instead delegates cross-node
\dtype{fp32} summation to the destination node's NVSwitch via
\code{multimem.\allowbreak ld\_reduce.\allowbreak relaxed.\allowbreak
sys.\allowbreak global.\allowbreak add.\allowbreak v4.f32}.

\paragraph{Latin-square sub-slice placement.}
For one \code{multimem.ld\_reduce} to sum partials from several
contributor nodes, they must sit at the \emph{same symmetric offset} on
distinct local GPUs $j \in [0, L)$ of destination node $m$, while rail
alignment lets local rank $j$ on a sender reach only local rank $j$ on
the receiver.

The idea is to trade contributor identity for GPU placement. Fix one
sub-slice: if every contributor of a wave places its partial at the
same offset on a \emph{different} destination GPU, one reduction
through the multicast pointer sums them all. The only problem is the
assignment: no two contributors in a wave may land on the same GPU,
yet each sender must still use its own rail. Rotating the assignment
with the contributor index---a Latin square, shown for $L=4$ in
\Cref{fig:rs-architecture}b---satisfies both constraints at once, as
\Cref{prop:latin} proves.

\system satisfies both constraints with a Latin-square sub-slice
rotation (\Cref{fig:rs-architecture}b). Let $E_{\mathrm{stg}}$ be
the aligned per-rank staging span. For the one-pass schedule, we split it
into $L$ equal-capacity sub-slices $s\in[0,L)$, each with vector-aligned
capacity $E_{\mathrm{sub}}\ge\lceil E_{\mathrm{stg}}/L\rceil$; the final
live range may be shorter.
For destination node $m$ and contributor node $n\in[0,M)$, define the
send-order contributor index:
\begin{equation}
  c(m, n) =
  \begin{cases}
    (m - n) \bmod M, & \text{if } n \ne m \;\; (1 \le c < M), \\
    M,               & \text{if } n = m \;\; (\text{own node}).
  \end{cases}
  \label{eq:contributor-idx}
\end{equation}
Write $\bar c = c \bmod M$, so the own-node contributor has
$\bar c=0$.
In step $c \in \{1,\dots,M-1\}$, sender node $n$ targets
$t=(n+c)\bmod M$. Local rank $j\in[0,L)$ selects sub-slice
\begin{equation}
  s \;=\; (j + \bar c) \bmod L
  \label{eq:rs-block}
\end{equation}
and performs Level~1 reduction for all $L$ ranks on destination node
$t$. It stores the \dtype{fp32} partials in a
contiguous outbox block with $L$ rows of $E_{\mathrm{sub}}$ elements,
then transfers it over rail $j$ to inbox block $s$ on rank $(t,j)$.
Equivalently, from the perspective of a fixed destination node $m$ and a
fixed sub-slice $s \in [0, L)$, contributor $c$'s partial lands in inbox
block $s$ of destination local GPU:
\begin{equation}
  j(s, c) \;=\; (s - \bar c) \bmod L.
  \label{eq:latin-map}
\end{equation}

For a fixed $s$, the rotation sends each contributor in a wave to a
distinct destination GPU at the same symmetric offset, turning
contributor identity into GPU placement so that one NVSwitch reduction
folds up to $L$ cross-node partials. Rail transfers use no SM kernel;
only a short landing kernel runs after each wave arrives.

\begin{proposition}[Multi-Wave Latin-Square Collision-Freedom]
\label{prop:latin}
Let $W_{\mathrm{wave}}$ be the maximum contributors per wave and
$w\in[0,W)$ the wave index. When the $M$ contributors
$c \in \{1, \dots, M\}$ are partitioned in
increasing order into $W = \lceil M / W_{\mathrm{wave}} \rceil$ waves of
width $W_{\mathrm{wave}}$ dividing $L$ (in particular, the default
$W_{\mathrm{wave}} = L$), no two contributors in the same wave $w$ map to
the same local GPU $j(s, c)$ on destination node $m$ for any sub-slice
$s \in [0, L)$.
\end{proposition}

\Cref{app:latin-proof} gives the proof and exhaustive planner
validation.

\paragraph{Bounded waves and safe reuse.}
When $M>L$, contributors are split into waves of at most $L$ nodes. The
inbox has $L$ symmetric blocks per buffer generation, with one
generation for a single-wave schedule and two alternating generations
otherwise---an inbox double buffer---so its size does not grow with
$M$. Once a wave arrives, a fixed-size landing kernel
reads the aligned blocks through the multicast pointer and accumulates
their \dtype{fp32} partials into the output. Blocks written by the
previous wave on the same buffer generation but not the current one
are zeroed before
landing, and an own-node partial that would overwrite a buffer
generation still
being read goes to a spare block until that generation is free.

\paragraph{Wave-aligned network rounds.}
\system confines each network round to one wave, preserving the bounded
inbox while exposing up to $k$ concurrent rail transfers. Very large
calls may repeat the same mapping over multiple passes without changing
the buffer-generation or stale-block rules. \Cref{sec:two-phase,sec:solver} describe
pipelining and schedule selection.

\subsection{Numerical Semantics and Determinism}
\label{sec:rs-numerics}

\system-RS changes the association of the sum relative to stock FSDP2
but never narrows an intermediate partial. Level~1 adds the $L-1$ staged
contributions in local-rank order, then the own contribution, in
\dtype{fp32} (\code{add.rn.f32.bf16} rounds once on SM100+); Level~2
adds \dtype{fp32} node partials in the NVSwitch. A plan fixes the
topology, rank placement, and $(M,L,W_{\mathrm{wave}})$. Repeated
executions of a plan are bitwise identical in outputs, losses, and
parameters, as are runs that change tuning knobs that preserve operand
grouping: CTA counts, chunk sizes, passes, staging slots, host- versus
device-initiated execution, and In-Place versus Slot-Reuse. \system thus
adds no nondeterminism of its own
(\Cref{sec:eval-correctness}). Standard-normal \dtype{bf16} and
\dtype{fp32} inputs match a distributed \dtype{fp64} reference within
$10^{-5}\max(1,\lVert y_{\mathrm{ref}}\rVert_\infty)$; integer inputs
with exactly representable sums match stock FSDP2 bit for bit; and NaN
and $\pm\infty$ inputs produce the expected output classes.

\section{Pipelining and Memory Efficiency}
\label{sec:pipeline}

Low-SM collectives must still become training-level overlap with
bounded memory. \system separates gradient reads from network progress
and makes gradient residence an explicit choice: In-Place keeps
gradients in autograd storage and prefetches the next all-gather early,
while Slot-Reuse moves them into a vacated all-gather slot and frees
their storage immediately.

\subsection{Scheduling and Allocator Constraints}
\label{sec:allocator-paradox}

Let $\ell$ index layers in forward order and $e$ successive
reduce-scatter calls. When \system-RS reads autograd's gradients
$\{G_p\}$ in place, they stay live until its Level~1 pushes and sums
finish; FSDP2 keeps a collective's input until the next group's
\code{post\_backward} observes its completion event, which is therefore
the default release point. Three constraints shape the schedule.

\paragraph{Read--network coupling.}
With one shared outbox bank, call $e$ cannot start reading until call
$e-1$'s sends release the bank, and forcing the compute stream to wait
for reads puts their latency on the critical path. Local reads therefore
need independent temporary state.

\paragraph{The host-time allocator paradox.}
Asynchronous reclamation via \code{record\_stream} or a polling thread
is GPU-safe, but PyTorch's caching allocator assigns blocks when the CPU
enqueues work, one or two layers ahead of the GPU. By the time an event
reports layer $\ell$ safe to free, later layers have already claimed
fresh segments (\Cref{fig:host-time}, \Cref{app:pipeline-details}), so
reclamation must happen at an allocator-visible host callback.

\paragraph{Backward prefetch inversion.}
FSDP2's \code{post\_backward} enqueues the current reduce-scatter before
the next \code{pre\_backward} enqueues the next all-gather, which can
then wait behind a whole multi-node network phase. Whether \system can
move it ahead depends on the mode: In-Place leaves the landing slot
free, whereas Slot-Reuse stores gradients there until reads complete.

\begin{figure}[t]
  \centering
\begin{tikzpicture}[td, x=1pt, y=1pt]
\def\tdOp#1#2#3#4#5{
  \filldraw[#4] ({50+0.24*#1},#3) rectangle ({50+0.24*#2},#3-8.5);
  \node[font=\tdtiny, inner sep=0pt] at ({50+0.12*(#1+#2)},#3-4.25) {#5};}
\def\tdBar#1#2#3#4{
  \filldraw[#4] ({50+0.24*#1},#3) rectangle ({50+0.24*#2},#3-3.5);}
\def\tdRel#1#2{
  \draw[draw=tdINK, line width=0.8pt] ({50+0.24*#1},#2+1) -- ({50+0.24*#1},#2-8.8);}
\tikzset{
  tdcmp/.style={draw=tdINK!45, line width=0.3pt, fill=tdCMP},
  tdag/.style={draw=tdCE, line width=0.35pt, fill=tdCEf},
  tdrd/.style={draw=tdSMd, line width=0.35pt, fill=tdSMf},
  tdrn/.style={draw=tdNET, line width=0.35pt, fill=tdNETf},
  tdgr/.style={draw=tdINK!55, line width=0.3pt, fill=tdCMPf},
}
\foreach \tdy/\tdname in {-15/In-Place, -90/Slot-Reuse}{%
  \node[font=\tdbold, anchor=base west, inner sep=0pt] at (0,\tdy) {\tdname};
  \foreach \tdk/\tdlane in {0/compute, 1/AG, 2/RS Level 1, 3/RS Level 2}{%
    \node[font=\tdsmall, anchor=east, inner sep=1pt] at (46,\tdy-9.25-11*\tdk) {\tdlane};}
  \node[font=\tdsmall, anchor=east, inner sep=1pt] at (46,\tdy-52.9) {gradients};}

\def\tdT{-20}
\tdOp{0}{240}{\tdT}{tdcmp}{bwd L7}\tdOp{240}{480}{\tdT}{tdcmp}{bwd L6}\tdOp{480}{720}{\tdT}{tdcmp}{bwd L5}
\tdOp{242}{338}{\tdT-11}{tdag}{AG L5}\tdOp{482}{578}{\tdT-11}{tdag}{AG L4}
\tdOp{242}{382}{\tdT-22}{tdrd}{read L7}\tdOp{482}{622}{\tdT-22}{tdrd}{read L6}
\tdOp{338}{488}{\tdT-33}{tdrn}{net L7}\tdOp{578}{728}{\tdT-33}{tdrn}{net L6}
\tdBar{170}{480}{\tdT-44}{tdgr}\tdBar{650}{780}{\tdT-44}{tdgr}
\tdBar{410}{720}{\tdT-48.3}{tdgr}
\tdRel{480}{\tdT-44}\tdRel{720}{\tdT-44}
\node[font=\tdtiny, anchor=east, inner sep=1pt] at ({50+0.24*170},\tdT-45.75) {L7};
\node[font=\tdtiny, anchor=east, inner sep=1pt] at ({50+0.24*410},\tdT-50.05) {L6};
\draw[decorate, decoration={brace, mirror, amplitude=1.8pt}, line width=0.4pt, tdTAX]
  ({50+0.24*410},\tdT-53.3) -- ({50+0.24*480},\tdT-53.3);
\node[font=\tdtiny, text=tdTAX, anchor=north, inner sep=1pt] at ({50+0.24*445},\tdT-55.3) {two groups live};

\def\tdT{-95}
\tdOp{0}{240}{\tdT}{tdcmp}{bwd L7}\tdOp{266}{506}{\tdT}{tdcmp}{bwd L6}\tdOp{532}{772}{\tdT}{tdcmp}{bwd L5}
\tdOp{240}{266}{\tdT}{tdag}{}\tdOp{506}{532}{\tdT}{tdag}{}
\tdOp{410}{506}{\tdT-11}{tdag}{AG L5}\tdOp{676}{772}{\tdT-11}{tdag}{AG L4}
\tdOp{268}{408}{\tdT-22}{tdrd}{read L7}\tdOp{534}{674}{\tdT-22}{tdrd}{read L6}
\tdOp{278}{498}{\tdT-33}{tdrn}{net L7}\tdOp{544}{764}{\tdT-33}{tdrn}{net L6}
\foreach \tdt in {408,674}{%
  \draw[draw=tdINK, line width=0.45pt, ->] ({50+0.24*\tdt+0.6},\tdT-22) -- ({50+0.24*\tdt+0.6},\tdT-19.7);}
\tdBar{170}{266}{\tdT-44}{tdgr}\tdBar{436}{532}{\tdT-44}{tdgr}\tdBar{702}{772}{\tdT-44}{tdgr}
\tdBar{240}{408}{\tdT-48.3}{tdag}\tdBar{506}{674}{\tdT-48.3}{tdag}
\tdRel{266}{\tdT-44}\tdRel{532}{\tdT-44}
\node[font=\tdtiny, anchor=east, inner sep=1pt] at ({50+0.24*170},\tdT-45.75) {L7};
\node[font=\tdtiny, text=tdCE, anchor=north, inner sep=1pt] at ({50+0.24*324},\tdT-52.8) {in AG slot};
\node[font=\tdtiny, text=tdMUTE, anchor=north, inner sep=0.5pt] at ({50+0.24*253},\tdT-8.7) {copy};

\begin{scope}[shift={(0,-167)}, font=\tdsmall]
  \draw[draw=tdINK, line width=0.8pt] (0,3) -- (0,-3);
  \node[anchor=west, inner sep=1.5pt] at (1,0) {autograd storage released};
  \draw[draw=tdINK, line width=0.45pt, ->] (112,-3) -- (112,3);
  \node[anchor=west, inner sep=1.5pt] at (113,0) {gated by read completion};
\end{scope}
\end{tikzpicture}
  \caption{Backward scheduling in the two modes. In-Place lets the next
  all-gather run during reduce-scatter but briefly keeps two gradient
  groups live; Slot-Reuse copies gradients into a vacated all-gather
  slot and gates the slot's reuse on read completion.}
  \Description{Stream timelines of the two modes over three backward
  layers, showing compute, all-gather, local reads, network phases, and
  gradient lifetimes.}
  \label{fig:pipeline-modes}
\end{figure}
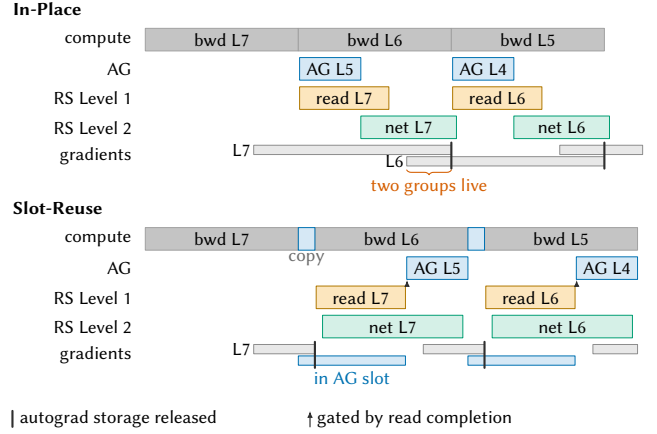

\subsection{Two-Phase Read--Network Pipelining}
\label{sec:two-phase}

For call $e$, \system splits source consumption from inter-node progress
and assigns its temporary state to bank $b=e\bmod 2$
(\Cref{fig:pipeline-modes}):
\begin{itemize}
  \item \textbf{Read phase.} As soon as backward produces $\{G_p\}$,
  Level~1 consumes them into an outbox bank and own-node spares without
  waiting for the previous call's network phase; a read-completion event
  marks when the source is no longer needed.
  \item \textbf{Network phase.} In collective order, rail rounds send
  ready outbox blocks and trigger the Level~2 wave landings
  (\Cref{sec:rs-level2}). A full outbox has a block for every remote
  output, so all reads proceed without waiting for the network; a
  partial outbox recycles blocks after their sends, overlapping the
  remaining reads with network progress.
\end{itemize}
Read completion thus controls source lifetime, per-round progress
controls block reuse, and final completion preserves collective order.

\paragraph{Safety and progress.}
The schedule enforces three happens-before relations: gradients and slot
copies stay live until read completion; an outbox block or inbox buffer
generation is reused only after its previous send or landing; and every wait
targets a producer already enqueued in the same collective order on all
ranks. With \Cref{prop:latin}, these give reclamation safety, buffer
safety, and deadlock freedom; \Cref{app:pipeline-details} details the
event protocol and fault-injection checks.

\subsection{In-Place and Slot-Reuse Modes}
\label{sec:early-prefetch}
\label{sec:two-modes}

The two modes choose where gradients reside while \system-RS consumes
them (\Cref{fig:pipeline-modes}). Both preserve the parameter-major
layout and avoid packing; they differ in gradient lifetime, all-gather
slot availability, and the number $n_{\mathrm{out}}$ of outbox blocks
per bank. Let $\mathrm{RS}_{\mathrm{net}}(\ell)$ denote layer $\ell$'s
network phase.

\paragraph{In-Place.}
\system reads $\{G_p\}$ directly. For the default one-pass plan, a full
outbox holds all $M-1$ remote destinations ($n_{\mathrm{out}}=M-1$), so
read completion does not wait
for the NIC; FSDP2 frees layer $\ell$'s gradients at the next group's
\code{post\_backward}, so two gradient groups are briefly live. Once
layer $\ell$ has resharded and released its ring slot, \system issues
the next backward all-gather ahead of $\mathrm{RS}_{\mathrm{net}}(\ell)$,
tracking forward invocation order so that repeated checkpointed or
accumulated invocations stay distinct.

\paragraph{Slot-Reuse.}
At \code{post\_backward}($\ell$), the layer has vacated an all-gather
slot with the same byte geometry as $\{G_p\}$. \system copies the
gradients into it, keeping their layout, and frees their storage at the
current host callback. The slot is now the reduce-scatter source, so its
next all-gather waits for read completion instead of early prefetch.
Since memory is already reclaimed, reads can follow network progress
through a partial outbox that covers two consecutive rounds. With $k_i
\le k$ transfers in round $i$, the per-bank depth is
\begin{equation}
  n_{\mathrm{out}} \;=\; \max_i \bigl(k_i + k_{i-1}\bigr) \;\le\; 2k,
  \qquad k_{-1}=0.
  \label{eq:partial-outbox}
\end{equation}
\subsection{Bounded Memory Analysis}
\label{sec:memory-model}

Within \system's arena, the two modes differ only in outbox depth. We
bound the per-GPU temporary RS staging $\mathrm{Area}_{\mathrm{RS}}$ for
our target configuration ($L=8$, \dtype{bf16} inputs, one-pass
pipeline), excluding the output shard, autograd gradients, and the shared
all-gather ring and receive stage. Let $S$ be one rank's \dtype{bf16}
shard size in the largest group ($S_{\mathrm{fp32}}=2S$),
$n_{\mathrm{stg}}$ the local staging slots, and
$n_{\mathrm{par}}\in\{1,2\}$ the number of inbox buffer generations. Then
\begin{equation}
  \mathrm{Area}_{\mathrm{RS}}
  \;\le\;
  \bigl(4n_{\mathrm{stg}} + 8n_{\mathrm{par}} + 2n_{\mathrm{out}} + 2\bigr)
  S_{\mathrm{fp32}}.
  \label{eq:area-general}
\end{equation}
since each \dtype{bf16} staging slot takes at most $4S_{\mathrm{fp32}}$,
each buffer generation eight blocks, the two outbox banks $2n_{\mathrm{out}}$
blocks, and each bank one own-node spare. With the default
$n_{\mathrm{stg}}=2$,
\begin{align}
  \mathrm{Area}_{\mathrm{InPlace}} &\le \bigl(2M + 8 n_{\mathrm{par}} + 8\bigr) S_{\mathrm{fp32}}
  \;=\; \left(\frac{1}{2} + \frac{2 n_{\mathrm{par}} + 2}{M}\right) NS, \label{eq:area-eff} \\
  \mathrm{Area}_{\mathrm{Slot}} &\le \bigl(4k + 8 n_{\mathrm{par}} + 10\bigr) S_{\mathrm{fp32}}
  \;=\; \frac{2k + 4 n_{\mathrm{par}} + 5}{2 M} \, NS. \label{eq:area-thr}
\end{align}
At fixed $k$, In-Place approaches $0.5NS$ as $M$ grows, whereas
Slot-Reuse decays as $O(k/M)$. At $M=64$, $n_{\mathrm{par}}=2$, and
$k=4$, the bounds are $0.594NS$ and $0.164NS$ (fixed-size chunks lower
In-Place to $0.56NS$).

During a layer's backward, FSDP2 holds $4NS$: the layer's unsharded
parameters, the prefetched all-gather output (allocated on another
stream, so the reduce-scatter never reuses its memory), and the previous
layer's \dtype{fp32} reduce-scatter input ($2NS$). \system's two ring
slots and receive stage ($(2+1/L)NS$) replace the first two, so with the
staging above and, for In-Place, the previous layer's \dtype{bf16}
gradients ($NS$; Slot-Reuse keeps them in a vacated slot), \system holds
at most $3.72NS$ under In-Place and $2.29NS$ under Slot-Reuse. \Cref{app:adapter-details} covers the adapter's storage and
expert-subgroup details.

\section{Implementation and Runtime}
\label{sec:runtime}

\system comprises about 5,100 lines of C++/CUDA, 2,300 lines of
Python, and 1,800 lines of tests. It uses only public NCCL~2.31 APIs
(symmetric windows, local peer pointers, multicast bindings, and RMA
puts and signals), so it runs on unmodified NCCL; all reported runs use
stock NCCL~2.31.2. The FSDP2 adapter replaces FSDP2's internal
all-gather, copy-out, reduce-scatter, unsharded-parameter allocation,
and post-forward (and, for In-Place, backward-prefetch) hooks after
checking their signatures; model code and compute kernels are unchanged.
Neither collective copies metadata to the GPU per call: the
reduce-scatter passes gradient pointers as kernel parameters and reads a
segment table uploaded at registration. The control plane must scale too:
\system bounds connection state with
topology-scoped communicators, host submission with a small stream pool,
and SM and memory use with plan selection.

\subsection{Topology-Scoped Communicators}
\label{sec:rail-comm}

A single symmetric-window RMA communicator scales poorly. In NCCL~2.31's
host RMA proxy, the first \code{ncclComm\allowbreak Window\allowbreak
Register} opens a send and a receive RDMA queue pair (QP) from each of
$n_{\mathrm{ctx}}$ RMA contexts to every rank and runs an $N$-rank ring
barrier after each connection: $O(n_{\mathrm{ctx}}N^2)$ barrier steps
and $2n_{\mathrm{ctx}}N$ QPs per GPU. Yet rank $r = nL + l$ only ever
sends to its $M-1$ same-rail peers, so \system splits the domains:
\begin{enumerate}
  \item \textbf{Node-local domain.} The main $N$-rank communicator has no
  RMA contexts but keeps local symmetric addressing (LSA) peer pointers
  and NVSwitch multicast.
  \item \textbf{Rail domains.} \code{ncclCommSplit} forms $L$ disjoint
  $M$-rank rail communicators with $n_{\mathrm{ctx}}$ contexts each
  (four in our evaluation).
  \item \textbf{Shared storage.} Each arena is registered in both
  domains; the main handle serves LSA and multicast, the rail handle RMA.
\end{enumerate}
Both communicators cap their unused NCCL channels at one. The split cuts
per-GPU QPs by $L\times$ and barrier work by $L^2\times$ ($8\times$ and
$64\times$ at $L=8$). On 8 nodes, with lazy connection, communicator
state takes \textbf{0.21~GiB} per GPU versus 6.95~GiB for eager 32-channel
initialization, and the rail split plus RMA connection takes
\textbf{1.3--1.7~s}.

\subsection{Bounded Host Execution}
\label{sec:gin-engine}

Long copy-descriptor chains can block host submission behind queued GPU
work. \system spreads descriptor batches over a small pool of internal
streams, with device-side slot counters preserving order, which keeps
median enqueue latency at 0.62~ms for \system-AG and 0.80--0.82~ms for
\system-RS across tested queue depths. The network phase is host-driven
by default, and in training this host work stays hidden behind backward
compute. Where NCCL's GPU-initiated networking (GIN) runs over GDAKI, an
optional device path replaces per-round host calls with a 1-CTA kernel
per wave (\Cref{app:gin-solver}); all reported results use the host
path.

\subsection{Plan Selection and Validation}
\label{sec:solver}

Before the first registration, a calibrated first-order model picks the
sum and landing CTA budgets, round width $k$, chunk size, pass count,
and outbox depth from measured NIC, local-sum, landing, and per-round
wait rates: it provisions enough local-sum throughput to feed the rail,
bounds the landing tail, and widens rounds only when that amortizes peer
waits (\Cref{app:gin-solver}). On our eight-node system, the solver
selects $k=4$ and four sum CTAs. Because overlap tightens the
gradient-read deadline, the training configuration raises the sum
budget to 16
(\Cref{sec:eval-setup}), which changes neither placement nor arithmetic.
At installation, all ranks agree on
tensor signatures, options, and capabilities; unsupported groups fall
back to stock collectives, and inconsistent global capabilities abort.
A 239-configuration test suite covers bitwise and tuning invariance,
non-finite inputs, HSDP and expert subgroups, EMA/checkpoint paths, and
41 fault-injection negative tests over 11 fault kinds
(\Cref{app:pipeline-details,app:adapter-details}).

\section{Evaluation}
\label{sec:eval}

We ask how much step time and memory \system saves in training
(\Cref{sec:eval-e2e}), where the saving comes from
(\Cref{sec:eval-attribution}), and how the collectives scale
(\Cref{sec:eval-scaling}).

\subsection{Experimental Setup}
\label{sec:eval-setup}

We use up to eight nodes, with eight GPUs per NVSwitch domain and one
RDMA rail per GPU; all training runs use 64 GPUs. We train dense
\emph{Llama-3.1 70B} at
8K--32K tokens/GPU, the MoE \emph{Qwen3-30B-A3B} at 16K--32K
tokens/GPU, and the high-noise expert of
\emph{Wan2.2-T2V-A14B}, a 14B dense diffusion transformer, at 17 frames
and local batch 1--4 (7.8K--31.2K tokens/GPU). We compare Stock FSDP2, FSDP2 with PyTorch
Symmetric Memory (FSDP2+SymmMem), and \system, using In-Place for Llama
and Qwen3 and Slot-Reuse for Wan; \system-RS runs 16 sum and 4 landing CTAs with
$k=4$. All systems share the model, precision, optimizer, activation
checkpointing, and inputs. We report median step times over 40--100
post-warm-up steps from repeated launches and pair launches by round;
memory is the whole-device peak.

\subsection{End-to-End Performance and Memory}
\label{sec:eval-e2e}

\begin{figure*}[t]
  \centering
  \includegraphics[width=\textwidth]{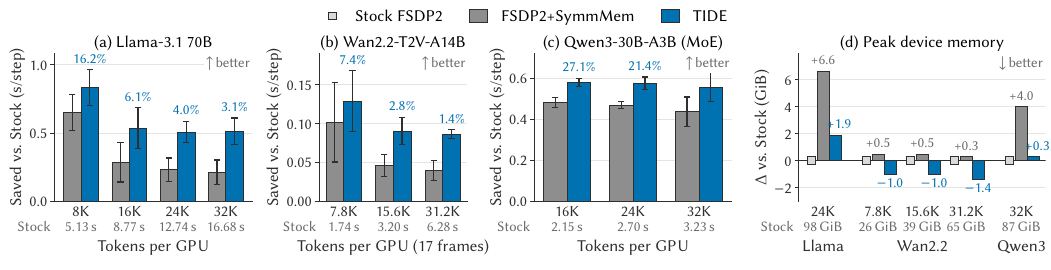}
  \caption{\system is fastest at every point and uses less memory than
  FSDP2+SymmMem (64 GPUs; Llama and Qwen3: In-Place; Wan: Slot-Reuse).
  (a--c) Time saved versus Stock; labels are the relative saving,
  Stock's step time is under each group, and error bars summarize
  run-to-run variation. (d) Peak memory versus Stock (repeats within
  0.1~GiB).}
  \Description{Grouped bars of the per-step time that FSDP2+SymmMem and
  TIDE save relative to Stock FSDP2 on Llama-3.1 70B at 8K to 32K tokens
  per GPU, on Wan2.2 at 7.8K to 31.2K tokens per GPU (17 frames), and
  on Qwen3-30B-A3B at 16K to 32K tokens per GPU, with Stock's step time
  under each group, and bars of peak-memory change relative to Stock
  annotated with Stock's absolute peak.}
  \label{fig:eval-e2e}
\end{figure*}

\system is the fastest system at every point
(\Cref{fig:eval-e2e}a--c). On Llama it cuts step time by 16.2\% at
8K tokens/GPU and 3.1\% at 32K versus Stock, and from 16K to 32K it
saves 0.51--0.54~s per step, 1.9--2.4$\times$ what FSDP2+SymmMem saves,
so the relative reduction shrinks only because the step grows. Wan shows
the same pattern on a shorter step: 7.4\%, 2.8\%, and 1.4\% versus Stock
at 7.8K, 15.6K, and 31.2K tokens/GPU, about twice what FSDP2+SymmMem
saves at the two larger sizes. For Llama and Wan, every paired-$t$ 95\%
interval against FSDP2+SymmMem excludes zero, and fixed-clock reruns of
Llama 24K and Wan at 15.6K tokens/GPU keep the same ordering.

Qwen3 shows the largest gains, as \Cref{eq:tax-scaling} predicts:
sparse MoE routing activates only a small fraction of the sharded
parameters per token, while every rank still gathers and reduces the
full parameter volume, so the recoverable overhead is a larger share
of the step. From 16K to 32K tokens/GPU, \system
cuts step time by 27.1\%, 21.4\%, and 17.3\% versus Stock and by
6.0\%, 4.8\%, and 4.3\% versus FSDP2+SymmMem
(\Cref{fig:eval-e2e}c).

\system uses less memory than FSDP2+SymmMem on all three models, and
Slot-Reuse uses less than Stock on Wan (\Cref{fig:eval-e2e}d). At Llama
24K, FSDP2+SymmMem holds the same live bytes as Stock but peaks 6.6~GiB
higher: its private pool reserves 9.4~GiB that no other tensor can
share, more than \system's whole 6.1~GiB arena, which also holds the
unsharded parameters. Stock takes the \dtype{fp32} reduce-scatter input
from the compute pool, where it shares memory with activations. In-Place
therefore peaks 4.7~GiB below FSDP2+SymmMem, adding 1.9~GiB over Stock,
mostly a communicator and an arena reserved once for the largest group
(Llama's vocabulary layer); its reduce-scatter staging does not grow
with the GPU count (\Cref{sec:eval-scaling}). On Wan, Slot-Reuse peaks
1.0--1.4~GiB below Stock and 1.5--1.7~GiB below FSDP2+SymmMem. In-Place
is only 0.3\% faster at batch~2 and uses 0.8~GiB more, so Wan uses
Slot-Reuse. On Qwen3 at 32K tokens/GPU, \system peaks 3.7~GiB below
FSDP2+SymmMem (87.3 vs.\ 91.0~GiB), within 0.3~GiB of Stock.

\subsection{Where the Gain Comes From}
\label{sec:eval-attribution}

\begin{figure*}[t]
  \centering
  \includegraphics[width=\textwidth]{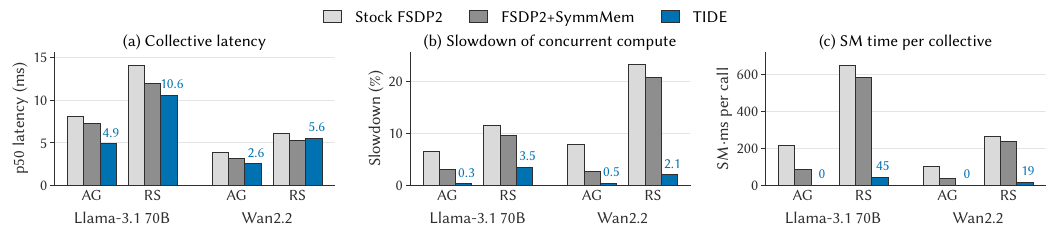}
  \caption{\system-AG barely slows concurrent compute; \system-RS cuts
  that slowdown several-fold (training shapes, fixed SM clock; lower
  is better; labels are \system's values). (a) Isolated calls.
  (b) MLP GEMM (Llama) or self-attention (Wan).}
  \Description{Grouped bars comparing Stock FSDP2, FSDP2+SymmMem, and
  TIDE on Llama-3.1 70B and Wan2.2 all-gather and reduce-scatter:
  median latency, slowdown of concurrent compute, and SM time per
  call.}
  \label{fig:eval-mechanism}
\end{figure*}

Both collectives contribute, and the gains add. Replacing one
collective at a time, in separate In-Place launches, returns about half
the combined saving: the all-gather alone is 50\% of it at Llama 24K and
55\% at Wan 15.6K, the reduce-scatter alone 58\% and 43\%, and the two
shares sum to 108\% and 98\%.

\system removes the layout copies and keeps communication off the SMs.
In a profile of one GPU during a Llama 8K step, Stock and FSDP2+SymmMem
spend 0.33 and 0.32~s per step on those copies. Communication drops
from 78~SM$\cdot$s per step under Stock and 34 under FSDP2+SymmMem to
3.8, and a communication kernel is resident for 8\% of the step instead
of 89\% and 59\%. \system's GEMMs are 1.3\% and 5.5\% shorter than under
the two baselines.

\system-AG is the fastest all-gather and barely slows concurrent
compute; \system-RS cuts that slowdown several-fold
(\Cref{fig:eval-mechanism}; \Cref{tab:eval-shape-replay} lists the
values). \system-AG uses no SMs, so concurrent compute slows by at most
0.5\%. FSDP2+SymmMem cuts all-gather SM time by 2.5--2.7$\times$ but
barely touches the reduce-scatter, whose packing and reduction still
run on SMs. \system-RS uses 12.5--14.5$\times$ less SM time
than either baseline, lowering the slowdown from 9.6--11.6\% to 3.5\%
on Llama and from 20.6--23.2\% to 2.1\% on Wan.

The advantage holds beyond the training shapes. Only \system's training
configuration finishes a 32~MiB reduce-scatter on 8 nodes faster than
Stock FSDP2 while slowing a one-wave GEMM by less than 4\%
(\Cref{app:microbench,tab:eval-overlap,fig:eval-cta}); against a large
GEMM and a memory-bound kernel, the training configuration adds 4.2\%
and 1.6\%, versus 13.6\% and 13.3\% for Stock FSDP2.

\subsection{Scaling and Resource Efficiency}
\label{sec:eval-scaling}

\begin{figure}[t]
  \centering
  \includegraphics[width=\columnwidth]{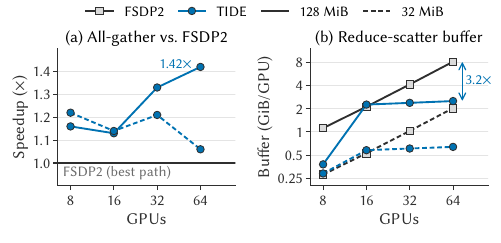}
  \caption{\system-AG stays faster from 8 to 64 GPUs; its
  reduce-scatter buffer stays flat
  (\Cref{tab:eval-ag,tab:eval-rs}). (a) Median speedup over FSDP2's
  fastest all-gather. (b) Buffer per GPU.}
  \Description{Two line plots over 8 to 64 GPUs. Left: TIDE all-gather
  speedup in median latency over the best FSDP2 path ranges from 1.06 to
  1.42 for 32 and 128 MiB shards; for 128 MiB it rises from 1.13 at 16
  GPUs to 1.42 at 64. Right, log scale: FSDP2's reduce-scatter buffer
  doubles with the GPU count, reaching 2.03 and 8.12 GiB at 64 GPUs,
  while TIDE's stays near 0.6 and 2.4 GiB from 16 GPUs on, a 3.2 times
  gap at 64 GPUs.}
  \label{fig:eval-scaling}
\end{figure}

\system-AG stays ahead of FSDP2's fastest all-gather, copy-out
included, and uses no SMs (\Cref{fig:eval-scaling}a). Across 1--8 nodes
and 32--128~MiB shards the median speedup is $1.06$--$1.42\times$.
On 128~MiB shards it grows from $1.13\times$ at 16 GPUs to $1.42\times$
at 64; 32~MiB shards stay between $1.06\times$ and $1.22\times$.

\system-RS keeps a flat buffer while FSDP2's grows with the GPU count
(\Cref{fig:eval-scaling}b). FSDP2 packs one \dtype{fp32} block per GPU,
each the size of the output shard, so a fixed output makes that buffer
linear in the GPU count. \system-RS stays within 18--21 such blocks from
two to eight nodes (0.58--0.64~GiB at 32~MiB outputs and
2.26--2.51~GiB at 128~MiB), and FSDP2's buffer is $3.2\times$ as large at 64
GPUs for both sizes ($0.9\times$ and $1.7\times$ at 16 and 32). For a fixed model,
FSDP2 packs $2NS$ at every scale; the 64-node analysis of
\Cref{sec:memory-model} bounds \system's staging at $0.59NS$ (In-Place)
and $0.16NS$ (Slot-Reuse), at most 1.2 and 0.3~GiB per GPU on Llama's
largest group against FSDP2's 3.9~GiB. On one node \system-RS is also
1.4--1.8$\times$ faster than both FSDP2 paths, and for 128~MiB outputs
on eight nodes its median latency is lower than both
(\Cref{tab:eval-rs}).

\subsection{Correctness and Robustness}
\label{sec:eval-correctness}

\system's training loss stays inside Stock FSDP2's own run-to-run
variation (\Cref{fig:convergence} in \Cref{app:convergence}).
Deterministic kernels stay off,
because they would change timing, so two Stock runs already diverge and
that spread is the reference. Over steps 500--1,000 of a 1,000-step run of Qwen3-30B-A3B, an MoE
model, on public C4, the 50-step moving averages differ from
a first Stock run by at most 0.73\% for \system and 1.12\% for a second
Stock run. Validation with EMA weight swaps, asynchronous checkpoint
save/resume, unsupported-group fallback, and peer-stall fail-stop also
pass. These runs are in addition to the operator checks in
\Cref{sec:rs-numerics}.

\section{Related Work}
\label{sec:related}

\paragraph{Sharded data parallelism and layout restructuring.}
ZeRO~\cite{rajbhandari2020zero}, ZeRO++~\cite{wang2023zeroplusplus},
MiCS~\cite{zhang2022mics}, PyTorch FSDP~\cite{zhao2023pytorchfsdp},
FSDP2 / TorchTitan~\cite{liang2025torchtitan}, and
SimpleFSDP~\cite{zhang2024simplefsdp} partition model states across
data-parallel ranks, while Megatron-LM~\cite{shoeybi2019megatron,korthikanti2023activation}
and DeepSpeed-MoE~\cite{rajbhandari2022deepspeedmoe} combine data, tensor,
sequence, and expert parallelism. FSDP2's per-parameter layout forces
the packing and copy-out of \Cref{sec:anatomy};
veScale-FSDP~\cite{vescalefsdp2025} and
Megatron-FSDP~\cite{megatronfsdp2025} remove them by restructuring model
storage around persistent flat or uneven collective buffers. \system
instead preserves FSDP2's per-parameter \code{DTensor} storage and
expresses non-contiguous placement through a byte table, analogous to
MPI derived datatypes~\cite{gropp1999using,thakur2005optimization}.

\paragraph{Symmetric memory and direct placement.}
NCCL 2.28--2.31~\cite{nccl231,ptx9} exposes symmetric windows,
copy-engine collectives, NVLink SHARP \code{multimem}, host RMA, and
device GIN.
Related symmetric-memory and direct-placement systems include PyTorch's
FSDP2 integration~\cite{pytorchsymmmem2025}, AMD
MORI~\cite{mori2025}, DeepEP~\cite{deepep2025}, and NCCL
PACE~\cite{pace2025}. PyTorch's symmetric FSDP2 path removes all-gather
transport SMs but keeps rank-major outputs and gradient packing, and its
multi-node copy-engine all-gather is unpaced, which can trigger RoCE
incast (\Cref{sec:ag-residence}). \system co-designs placement for both
collectives while preserving FSDP2's storage contract.

\paragraph{Collective libraries and GPU-initiated communication.}
NCCL~\cite{nccl231} is the de facto GPU collective library, and Hu et
al.~\cite{hu2025demystifying} dissect its protocols and algorithms.
Blink~\cite{wang2020blink}, TACCL~\cite{shah2023taccl}, and
MSCCLang~\cite{cowan2023msccl} generate or express topology-aware
collective algorithms, and MSCCL++~\cite{shah2025mscclpp} exposes
GPU-driven communication primitives. GPUDirect
Async~\cite{agostini2018gpudirect} lets the GPU trigger network
operations prepared by the CPU, and Unat et al.~\cite{unat2024landscape}
survey GPU-centric communication more broadly. \system adds no
programming model: it composes NCCL's public copy-engine,
symmetric-memory, \code{multimem}, and RMA primitives into collectives
whose output layout matches FSDP2's storage.

\paragraph{GPU sharing and interference.}
Co-running kernels contend for SMs, registers, shared memory, and memory
bandwidth. Chimera~\cite{park2015chimera} and Simultaneous
Multikernel~\cite{wang2016smk} add hardware preemption and fine-grained
sharing of each SM among kernels, while REEF~\cite{han2022reef},
Orion~\cite{strati2024orion}, and Tally~\cite{zhao2025tally} preempt,
schedule, or slice co-located DNN kernels in software to bound
interference. For concurrent computation and communication,
ConCCL~\cite{agrawal2024dma} moves all-gather and all-to-all, but not
reductions, onto GPU DMA engines, and FiCCO~\cite{ficco2025} explores
DMA-based finer-grain overlap. \system shares this goal for
FSDP2's collectives: it keeps transport off SMs, models
the remaining SM cost by cluster slots and residence time
(\Cref{sec:wave-quant}), and adds the direction-dependent HBM asymmetry
of \Cref{sec:direction-tax}.

\paragraph{Overlap and in-network reduction.}
CoCoNet~\cite{jangda2022coconet}, decomposition
pipelines~\cite{wang2023overlap}, Centauri~\cite{chen2024centauri},
T3~\cite{pati2024t3}, FLUX~\cite{chang2024flux},
TileLink~\cite{tilelink2025}, and FlashOverlap~\cite{flashoverlap2025}
fuse or schedule compute--communication overlap across general,
tensor-parallel, and MoE workloads. \system instead targets the FSDP2
layout boundary while leaving compute kernels unchanged. InfiniBand
SHARP~\cite{graham2016scalable}, SwitchML~\cite{sapio2021scaling}, and
ATP~\cite{lao2021atp} aggregate inside the inter-node network.
\system instead maps remote partials onto destination GPUs so a
node-local NVSwitch performs the final \dtype{fp32} reduction across
multi-node rail topologies, without programmable inter-node switch
support.

\section{Conclusion}
\label{sec:conclusion}

Layout conversion in FSDP2 is a collective problem, not a copy problem.
Solving it inside the collective lets \system keep FSDP2's
per-parameter storage, overlap gradient reads with the network through a
two-phase pipeline, and reuse all-gather slots as parameter storage and,
optionally, as reduce-scatter input.

On 64 GPUs training Llama-3.1 70B, Wan2.2, and Qwen3-30B-A3B, \system
is the fastest system at every measured point, cutting step time by up
to 27.1\% relative to Stock FSDP2 and 6.0\% relative to
FSDP2+SymmMem; its peak
memory is up to 4.7~GiB below FSDP2+SymmMem's and, with Slot-Reuse on
Wan2.2, below Stock FSDP2's. Its all-gather uses no SMs and is up to
$1.42\times$ faster than FSDP2's fastest all-gather plus copy-out, and
in training-shape replays its reduce-scatter cuts interference with
concurrent compute by $2.8\text{--}11\times$. Throughout, it preserves
\dtype{fp32} reduction and plan-scoped determinism.

\paragraph{Artifact availability.}
We intend to release \system's implementation, benchmark drivers, and
automated correctness gate upon publication.

\bibliographystyle{ACM-Reference-Format}
\bibliography{references}

@inproceedings{rajbhandari2020zero,
  title={{ZeRO}: Memory Optimizations Toward Training Trillion Parameter Models},
  author={Rajbhandari, Samyam and Rasley, Jeff and Ruwase, Olatunji and He, Yuxiong},
  booktitle={Proceedings of the International Conference for High Performance Computing, Networking, Storage and Analysis (SC)}, pages={1--16},
  publisher={IEEE}, address={Piscataway, NJ, USA}, year={2020}, doi={10.1109/SC41405.2020.00024}
}

@inproceedings{wang2023zeroplusplus,
  title={{ZeRO++}: Extremely Efficient Collective Communication for Large Model Training},
  author={Wang, Guanhua and Qin, Heyang and Jacobs, Sam Ade and Wu, Xiaoxia and Holmes, Connor and Yao, Zhewei and Rajbhandari, Samyam and Ruwase, Olatunji and Yan, Feng and Yang, Lei and He, Yuxiong},
  booktitle={International Conference on Learning Representations (ICLR)}, numpages={12},
  publisher={OpenReview.net}, address={Amherst, MA, USA}, year={2024}, url={https://openreview.net/forum?id=gx2BT0a9MQ}
}

@article{zhang2022mics,
  title={{MiCS}: Near-linear Scaling for Training Gigantic Model on Public Cloud},
  author={Zhang, Zhen and Zheng, Shuai and Wang, Yida and Chiu, Justin and Karypis, George and Chilimbi, Trishul and Li, Mu and Jin, Xin},
  journal={Proceedings of the VLDB Endowment}, volume={16}, number={1}, pages={37--50}, year={2022}, doi={10.14778/3561261.3561265}
}

@article{zhao2023pytorchfsdp,
  title={{PyTorch FSDP}: Experiences on Scaling Fully Sharded Data Parallel},
  author={Zhao, Yanli and Gu, Andrew and Varma, Rohan and Luo, Liang and Huang, Chien-Chin and Xu, Min and Wright, Less and Shojanazeri, Hamid and Ott, Myle and Shleifer, Sam and Desmaison, Alban and Balioglu, Can and Damania, Pritam and Nguyen, Bernard and Chauhan, Geeta and Hao, Yuchen and Mathews, Ajit and Li, Shen},
  journal={Proceedings of the VLDB Endowment}, volume={16}, number={12}, pages={3848--3860}, year={2023}, doi={10.14778/3611540.3611569}
}

@misc{zhang2024simplefsdp,
  title={{SimpleFSDP}: Simpler Fully Sharded Data Parallel with torch.compile},
  author={Zhang, Ruisi and Liu, Tianyu and Feng, Will and Gu, Andrew and Purandare, Sanket and Liang, Wanchao and Massa, Francisco},
  year={2024}, eprint={2411.00284}, archivePrefix={arXiv}, url={https://arxiv.org/abs/2411.00284}
}

@inproceedings{vescalefsdp2025,
  title={{veScale-FSDP}: Flexible and High-Performance {FSDP} at Scale},
  author={Wang, Zezhou and Li, Youjie and Lin, Zhiqi and Yang, Jiacheng and Xie, Cong and Feng, Guanyu and Zhong, Zheng and Huang, Ziyue and Zhu, Hongyu and Zhang, Zhi and Peng, Yanghua and Liu, Xin},
  booktitle={Proceedings of the 9th Conference on Machine Learning and Systems (MLSys), Industry Track}, numpages={14},
  publisher={MLSys}, address={Bellevue, WA, USA}, year={2026},
  url={https://proceedings.mlsys.org/paper_files/paper/2026/hash/d5cd70b708f726737e2ebace18c3f71b-Abstract-Conference.html}
}

@misc{megatronfsdp2025,
  title={{Megatron-FSDP}: High-Performance Fully Sharded Data Parallelism in {Megatron-Core}},
  author={{NVIDIA Corporation}}, year={2025}, note={NVIDIA Megatron-Core Documentation},
  howpublished={\url{https://docs.nvidia.com/megatron-core/}}
}

@inproceedings{paszke2019pytorch,
  title={{PyTorch}: An Imperative Style, High-Performance Deep Learning Library},
  author={Paszke, Adam and Gross, Sam and Massa, Francisco and Lerer, Adam and Bradbury, James and Chanan, Gregory and Killeen, Trevor and Lin, Zeming and Gimelshein, Natalia and Antiga, Luca and Desmaison, Alban and K{\"o}pf, Andreas and Yang, Edward and DeVito, Zachary and Raison, Martin and Tejani, Alykhan and Chilamkurthy, Sasank and Steiner, Benoit and Fang, Lu and Bai, Junjie and Chintala, Soumith},
  booktitle={Advances in Neural Information Processing Systems (NeurIPS)}, pages={8024--8035},
  publisher={Curran Associates, Inc.}, address={Red Hook, NY, USA}, year={2019}
}

@inproceedings{ansel2024pytorch2,
  title={{PyTorch 2}: Faster Machine Learning Through Dynamic Python Bytecode Transformation and Graph Compilation},
  author={Ansel, Jason and Yang, Edward and He, Horace and Gimelshein, Natalia and Jain, Animesh and Voznesensky, Michael and Bao, Bin and Bell, Peter and Berard, David and Burovski, Evgeni and Chauhan, Geeta and Chourdia, Anjali and Constable, Will and Desmaison, Alban and DeVito, Zachary and Ellison, Elias and Feng, Will and Gong, Jiong and Gschwind, Michael and Hirsh, Brian and Huang, Sherlock and Kalambarkar, Kshiteej and Kirsch, Laurent and Lazos, Michael and Lezcano, Mario and Liang, Yanbo and Liang, Jason and Lu, Yinghai and Luk, C. K. and Maher, Bert and Pan, Yunjie and Puhrsch, Christian and Reso, Matthias and Saroufim, Mark and Siraichi, Marcos Yukio and Suk, Helen and Zhang, Shunting and Suo, Michael and Tillet, Phil and Zhao, Xu and Wang, Eikan and Zhou, Keren and Zou, Richard and Wang, Xiaodong and Mathews, Ajit and Wen, William and Chanan, Gregory and Wu, Peng and Chintala, Soumith},
  booktitle={Proceedings of the 29th ACM International Conference on Architectural Support for Programming Languages and Operating Systems, Volume 2}, pages={929--947},
  publisher={Association for Computing Machinery}, address={New York, NY, USA}, year={2024}, doi={10.1145/3620665.3640366}
}

@inproceedings{liang2025torchtitan,
  title={{TorchTitan}: One-stop {PyTorch} Native Solution for Production Ready {LLM} Pretraining},
  author={Liang, Wanchao and Liu, Tianyu and Wright, Less and Constable, Will and Gu, Andrew and Huang, Chien-Chin and Zhang, Iris and Feng, Wei and Huang, Howard and Wang, Junjie and Purandare, Sanket and Nadathur, Gokul and Idreos, Stratos},
  booktitle={International Conference on Learning Representations (ICLR)}, numpages={21},
  publisher={OpenReview.net}, address={Amherst, MA, USA}, year={2025}, url={https://openreview.net/forum?id=SFN6Wm7YBI}
}

@misc{shoeybi2019megatron,
  title={{Megatron-LM}: Training Multi-Billion Parameter Language Models Using Model Parallelism},
  author={Shoeybi, Mohammad and Patwary, Mostofa and Puri, Raul and LeGresley, Patrick and Casper, Jared and Catanzaro, Bryan},
  year={2019}, eprint={1909.08053}, archivePrefix={arXiv}, url={https://arxiv.org/abs/1909.08053}
}

@inproceedings{korthikanti2023activation,
  title={Reducing Activation Recomputation in Large Transformer Models},
  author={Korthikanti, Vijay Anand and Casper, Jared and Lym, Sangkug and McAfee, Lawrence and Andersch, Michael and Shoeybi, Mohammad and Catanzaro, Bryan},
  booktitle={Proceedings of Machine Learning and Systems (MLSys)}, pages={341--353},
  publisher={MLSys}, address={Miami, FL, USA}, year={2023}
}

@inproceedings{lepikhin2020gshard,
  title={{GShard}: Scaling Giant Models with Conditional Computation and Automatic Sharding},
  author={Lepikhin, Dmitry and Lee, HyoukJoong and Xu, Yuanzhong and Chen, Dehao and Firat, Orhan and Huang, Yanping and Krikun, Maxim and Shazeer, Noam and Chen, Zhifeng},
  booktitle={International Conference on Learning Representations (ICLR)}, numpages={35},
  publisher={OpenReview.net}, address={Amherst, MA, USA}, year={2021}, url={https://openreview.net/forum?id=qrwe7XHTmYb}
}

@article{fedus2022switch,
  title={Switch Transformers: Scaling to Trillion Parameter Models with Simple and Efficient Sparsity},
  author={Fedus, William and Zoph, Barret and Shazeer, Noam},
  journal={Journal of Machine Learning Research}, volume={23}, number={120}, pages={1--39}, year={2022}
}

@inproceedings{rajbhandari2022deepspeedmoe,
  title={{DeepSpeed-MoE}: Advancing Mixture-of-Experts Inference and Training to Power Next-Generation {AI} Scale},
  author={Rajbhandari, Samyam and Li, Conglong and Yao, Zhewei and Zhang, Minjia and Aminabadi, Reza Yazdani and Awan, Ammar Ahmad and Rasley, Jeff and He, Yuxiong},
  booktitle={Proceedings of the 39th International Conference on Machine Learning}, series={Proceedings of Machine Learning Research}, volume={162}, pages={18332--18346},
  publisher={PMLR}, address={Cambridge, MA, USA}, year={2022}, url={https://proceedings.mlr.press/v162/rajbhandari22a.html}
}

@book{gropp1999using,
  title={Using {MPI}: Portable Parallel Programming with the Message-Passing Interface},
  author={Gropp, William and Lusk, Ewing and Skjellum, Anthony},
  publisher={MIT Press}, address={Cambridge, MA, USA}, edition={2nd}, year={1999}
}

@article{bruck1997efficient,
  title={Efficient Algorithms for All-to-All Communications in Multiport Message-Passing Systems},
  author={Bruck, Jehoshua and Ho, Ching-Tien and Kipnis, Shlomo and Upfal, Eli and Weathersby, Derrick},
  journal={IEEE Transactions on Parallel and Distributed Systems}, volume={8}, number={11}, pages={1143--1156}, year={1997}, doi={10.1109/71.642949}
}

@article{thakur2005optimization,
  title={Optimization of Collective Communication Operations in {MPICH}},
  author={Thakur, Rajeev and Rabenseifner, Rolf and Gropp, William},
  journal={The International Journal of High Performance Computing Applications}, volume={19}, number={1}, pages={49--66}, year={2005}, doi={10.1177/1094342005051521}
}

@inproceedings{qian2024alibaba,
  title={Alibaba {HPN}: A Data Center Network for Large Language Model Training},
  author={Qian, Kun and Xi, Yongqing and Cao, Jiamin and Gao, Jiaqi and Xu, Yichi and Guan, Yu and Fu, Binzhang and Shi, Xuemei and Zhu, Fangbo and Miao, Rui and Wang, Chao and Wang, Peng and Zhang, Pengcheng and Zeng, Xianlong and Ruan, Eddie and Yao, Zhiping and Zhai, Ennan and Cai, Dennis},
  booktitle={Proceedings of the ACM SIGCOMM 2024 Conference}, pages={691--706},
  publisher={Association for Computing Machinery}, address={New York, NY, USA}, year={2024}, doi={10.1145/3651890.3672265}
}

@inproceedings{jiang2024megascale,
  title={{MegaScale}: Scaling Large Language Model Training to More Than 10,000 {GPUs}},
  author={Jiang, Ziheng and Lin, Haibin and Zhong, Yinmin and Huang, Qi and Chen, Yangrui and Zhang, Zhi and Peng, Yanghua and Li, Xiang and Xie, Cong and Nong, Shibiao and Jia, Yulu and He, Sun and Chen, Hongmin and Bai, Zhihao and Hou, Qi and Yan, Shipeng and Zhou, Ding and Sheng, Yiyao and Jiang, Zhuo and Xu, Haohan and Wei, Haoran and Zhang, Zhang and Nie, Pengfei and Zou, Leqi and Zhao, Sida and Xiang, Liang and Liu, Zherui and Li, Zhe and Jia, Xiaoying and Ye, Jianxi and Jin, Xin and Liu, Xin},
  booktitle={21st USENIX Symposium on Networked Systems Design and Implementation (NSDI)}, pages={745--760},
  publisher={USENIX Association}, address={Berkeley, CA, USA}, year={2024}, url={https://www.usenix.org/conference/nsdi24/presentation/jiang-ziheng}
}

@inproceedings{cowan2023msccl,
  title={{MSCCLang}: Microsoft Collective Communication Language},
  author={Cowan, Meghan and Maleki, Saeed and Musuvathi, Madanlal and Saarikivi, Olli and Xiong, Yifan},
  booktitle={Proceedings of the 28th ACM International Conference on Architectural Support for Programming Languages and Operating Systems, Volume 2}, pages={502--514},
  publisher={Association for Computing Machinery}, address={New York, NY, USA}, year={2023}, doi={10.1145/3575693.3575724}
}

@inproceedings{shah2025mscclpp,
  title={{MSCCL++}: Rethinking {GPU} Communication Abstractions for {AI} Inference},
  author={Hwang, Changho and Cheng, Peng and Dathathri, Roshan and Jangda, Abhinav and Maleki, Saeed and Musuvathi, Madan and Saarikivi, Olli and Shah, Aashaka and Yang, Ziyue and Li, Binyang and Rocha, Caio and Zhou, Qinghua and Ghazimirsaeed, Mahdieh and Anantharamu, Sreevatsa and Jose, Jithin},
  booktitle={Proceedings of the 31st ACM International Conference on Architectural Support for Programming Languages and Operating Systems, Volume 2}, pages={1201--1215},
  publisher={Association for Computing Machinery}, address={New York, NY, USA}, year={2026}, doi={10.1145/3779212.3790188}
}

@inproceedings{shah2023taccl,
  title={{TACCL}: Guiding Collective Algorithm Synthesis Using Communication Sketches},
  author={Shah, Aashaka and Chidambaram, Vijay and Cowan, Meghan and Maleki, Saeed and Musuvathi, Madan and Mytkowicz, Todd and Nelson, Jacob and Saarikivi, Olli and Singh, Rachee},
  booktitle={20th USENIX Symposium on Networked Systems Design and Implementation (NSDI)}, pages={593--612},
  publisher={USENIX Association}, address={Berkeley, CA, USA}, year={2023}, url={https://www.usenix.org/conference/nsdi23/presentation/shah}
}

@inproceedings{jangda2022coconet,
  title={Breaking the Computation and Communication Abstraction Barrier in Distributed Machine Learning Workloads},
  author={Jangda, Abhinav and Huang, Jun and Liu, Guodong and Sabet, Amir Hossein Nodehi and Maleki, Saeed and Miao, Youshan and Musuvathi, Madanlal and Mytkowicz, Todd and Saarikivi, Olli},
  booktitle={Proceedings of the 27th ACM International Conference on Architectural Support for Programming Languages and Operating Systems}, pages={402--416},
  publisher={Association for Computing Machinery}, address={New York, NY, USA}, year={2022}, doi={10.1145/3503222.3507778}
}

@inproceedings{wang2023overlap,
  title={Overlap Communication with Dependent Computation via Decomposition in Large Deep Learning Models},
  author={Wang, Shibo and Wei, Jinliang and Sabne, Amit and Davis, Andy and Ilbeyi, Berkin and Hechtman, Blake and Chen, Dehao and Murthy, Karthik Srinivasa and Maggioni, Marcello and Zhang, Qiao and Kumar, Sameer and Guo, Tongfei and Xu, Yuanzhong and Zhou, Zongwei},
  booktitle={Proceedings of the 28th ACM International Conference on Architectural Support for Programming Languages and Operating Systems, Volume 1}, pages={93--106},
  publisher={Association for Computing Machinery}, address={New York, NY, USA}, year={2023}, doi={10.1145/3567955.3567959}
}

@inproceedings{chen2024centauri,
  title={Centauri: Enabling Efficient Scheduling for Communication-Computation Overlap in Large Model Training via Communication Partitioning},
  author={Chen, Chang and Li, Xiuhong and Zhu, Qianchao and Duan, Jiangfei and Sun, Peng and Zhang, Xingcheng and Yang, Chao},
  booktitle={Proceedings of the 29th ACM International Conference on Architectural Support for Programming Languages and Operating Systems, Volume 3}, pages={178--191},
  publisher={Association for Computing Machinery}, address={New York, NY, USA}, year={2024}, doi={10.1145/3620666.3651379}
}

@inproceedings{pati2024t3,
  title={{T3}: Transparent Tracking \& Triggering for Fine-grained Overlap of Compute \& Collectives},
  author={Pati, Suchita and Aga, Shaizeen and Islam, Mahzabeen and Jayasena, Nuwan and Sinclair, Matthew D.},
  booktitle={Proceedings of the 29th ACM International Conference on Architectural Support for Programming Languages and Operating Systems, Volume 2}, pages={1146--1164},
  publisher={Association for Computing Machinery}, address={New York, NY, USA}, year={2024}, doi={10.1145/3620665.3640410}
}

@misc{chang2024flux,
  title={{FLUX}: Fast Software-based Communication Overlap on {GPUs} Through Kernel Fusion},
  author={Chang, Li-Wen and Bao, Wenlei and Hou, Qi and Jiang, Chengquan and Zheng, Ningxin and Zhong, Yinmin and Zhang, Xuanrun and Song, Zuquan and Yao, Chengji and Jiang, Ziheng and Lin, Haibin and Jin, Xin and Liu, Xin},
  year={2024}, eprint={2406.06858}, archivePrefix={arXiv}, url={https://arxiv.org/abs/2406.06858}
}

@inproceedings{tilelink2025,
  title={{TileLink}: Generating Efficient Compute-Communication Overlapping Kernels Using Tile-Centric Primitives},
  author={Zheng, Size and Fang, Jin and Zheng, Xuegui and Hou, Qi and Bao, Wenlei and Zheng, Ningxin and Jiang, Ziheng and Wang, Dongyang and Ye, Jianxi and Lin, Haibin and Chang, Li-Wen and Liu, Xin},
  booktitle={Proceedings of the 8th Conference on Machine Learning and Systems (MLSys)}, numpages={16},
  publisher={MLSys}, address={Santa Clara, CA, USA}, year={2025}
}

@inproceedings{flashoverlap2025,
  title={Efficient and Adaptable Overlapping for Computation and Communication via Signaling and Reordering},
  author={Hong, Ke and Li, Xiuhong and Liu, Minxu and Mao, Qiuli and Wu, Tianqi and Huang, Zixiao and Chen, Lufang and Wang, Zhong and Zhang, Yichong and Zhu, Zhenhua and Dai, Guohao and Wang, Yu},
  booktitle={Proceedings of the 21st European Conference on Computer Systems (EuroSys)}, pages={1894--1911},
  publisher={Association for Computing Machinery}, address={New York, NY, USA}, year={2026}, doi={10.1145/3767295.3769370}
}

@inproceedings{ficco2025,
  title={{FiCCO}: Design Space Exploration of {DMA}-Based Finer-Grain Compute Communication Overlap},
  author={Pal, Shagnik and Aga, Shaizeen and Pati, Suchita and Islam, Mahzabeen and John, Lizy K.},
  booktitle={2026 IEEE International Symposium on Workload Characterization (IISWC)}, numpages={15},
  publisher={IEEE}, address={Piscataway, NJ, USA}, year={2026}, doi={10.1109/IISWC71020.2026.00030}
}

@inproceedings{graham2016scalable,
  title={Scalable Hierarchical Aggregation Protocol ({SHArP}): A Hardware Architecture for Efficient Data Reduction},
  author={Graham, Richard L. and Koushnir, Vladimir and Levi, Lion and Margolin, Alex and Ronen, Tamir and Shpiner, Alexander and Wertheim, Oded and Zahavi, Eitan and Bureddy, Devendar and Lui, Pak and Rosenstock, Hal and Shainer, Gilad and Bloch, Gil and Goldenberg, Dror and Dubman, Mike and Kotchubievsky, Sasha},
  booktitle={2016 First International Workshop on Communication Optimizations in HPC (COMHPC)}, pages={1--10},
  publisher={IEEE}, address={Piscataway, NJ, USA}, year={2016}, doi={10.1109/COMHPC.2016.006}
}

@inproceedings{sapio2021scaling,
  title={Scaling Distributed Machine Learning with In-Network Aggregation},
  author={Sapio, Amedeo and Canini, Marco and Ho, Chen-Yu and Nelson, Jacob and Kalnis, Panos and Kim, Changhoon and Krishnamurthy, Arvind and Moshref, Masoud and Ports, Dan and Richtarik, Peter},
  booktitle={18th USENIX Symposium on Networked Systems Design and Implementation (NSDI)}, pages={785--808},
  publisher={USENIX Association}, address={Berkeley, CA, USA}, year={2021}, url={https://www.usenix.org/conference/nsdi21/presentation/sapio}
}

@inproceedings{lao2021atp,
  title={{ATP}: In-Network Aggregation for Multi-Tenant Learning},
  author={Lao, ChonLam and Le, Yanfang and Mahajan, Kshiteej and Chen, Yixi and Wu, Wenfei and Akella, Aditya and Swift, Michael},
  booktitle={18th USENIX Symposium on Networked Systems Design and Implementation (NSDI)}, pages={741--761},
  publisher={USENIX Association}, address={Berkeley, CA, USA}, year={2021}, url={https://www.usenix.org/conference/nsdi21/presentation/lao}
}

@misc{deepep2025,
  title={{DeepEP}: An Efficient Expert-Parallel Communication Library},
  author={Zhao, Chenggang and Zhou, Shangyan and Zhang, Liyue and Deng, Chengqi and Xu, Zhe and Liu, Yibo and Yu, Kuai and Li, Jiashi and Zhao, Liang},
  year={2025}, note={Open-source software repository}, howpublished={\url{https://github.com/deepseek-ai/DeepEP}}
}

@misc{pace2025,
  title={{PACE}: Parallelism-Aware Collective Engine}, author={{NVIDIA Corporation}}, year={2026},
  note={NCCL contrib software}, howpublished={\url{https://github.com/NVIDIA/nccl/tree/master/contrib/pace}}
}

@misc{mori2025,
  title={[{RFC}] {SDMA} {AllGather} and Zero-Copy Output for {PyTorch} {FSDP2}},
  author={{Advanced Micro Devices, Inc.}}, year={2026}, note={PyTorch RFC \#186601},
  howpublished={\url{https://github.com/pytorch/pytorch/issues/186601}}
}

@misc{pytorchsymmmem2025,
  title={Enable Copy Engine All-Gather in {FSDP}}, author={{PyTorch Contributors}}, year={2026},
  note={PyTorch Pull Request \#176613}, howpublished={\url{https://github.com/pytorch/pytorch/pull/176613}}
}

@misc{nccl231,
  title={{NVIDIA} Collective Communications Library ({NCCL}) v2.31 Documentation},
  author={{NVIDIA Corporation}}, year={2026}, note={NVIDIA Developer Documentation},
  howpublished={\url{https://docs.nvidia.com/deeplearning/nccl/}}
}

@misc{ptx9,
  title={Parallel Thread Execution {ISA} Version 9.0}, author={{NVIDIA Corporation}}, year={2025},
  note={CUDA Toolkit Documentation}, howpublished={\url{https://docs.nvidia.com/cuda/parallel-thread-execution/}}
}

@inproceedings{park2015chimera,
  title={Chimera: Collaborative Preemption for Multitasking on a Shared {GPU}},
  author={Park, Jason Jong Kyu and Park, Yongjun and Mahlke, Scott A.},
  booktitle={Proceedings of the Twentieth International Conference on Architectural Support for Programming Languages and Operating Systems}, pages={593--606},
  publisher={Association for Computing Machinery}, address={New York, NY, USA}, year={2015}, doi={10.1145/2694344.2694346}
}

@inproceedings{wang2016smk,
  title={Simultaneous Multikernel {GPU}: Multi-tasking Throughput Processors via Fine-Grained Sharing},
  author={Wang, Zhenning and Yang, Jun and Melhem, Rami G. and Childers, Bruce R. and Zhang, Youtao and Guo, Minyi},
  booktitle={2016 IEEE International Symposium on High Performance Computer Architecture (HPCA)}, pages={358--369},
  publisher={IEEE}, address={Piscataway, NJ, USA}, year={2016}, doi={10.1109/HPCA.2016.7446078}
}

@inproceedings{han2022reef,
  title={Microsecond-scale Preemption for Concurrent {GPU}-accelerated {DNN} Inferences},
  author={Han, Mingcong and Zhang, Hanze and Chen, Rong and Chen, Haibo},
  booktitle={16th USENIX Symposium on Operating Systems Design and Implementation (OSDI)}, pages={539--558},
  publisher={USENIX Association}, address={Berkeley, CA, USA}, year={2022}, url={https://www.usenix.org/conference/osdi22/presentation/han}
}

@inproceedings{strati2024orion,
  title={Orion: Interference-aware, Fine-grained {GPU} Sharing for {ML} Applications},
  author={Strati, Foteini and Ma, Xianzhe and Klimovic, Ana},
  booktitle={Proceedings of the Nineteenth European Conference on Computer Systems (EuroSys)}, pages={1075--1092},
  publisher={Association for Computing Machinery}, address={New York, NY, USA}, year={2024}, doi={10.1145/3627703.3629578}
}

@inproceedings{zhao2025tally,
  title={Tally: Non-Intrusive Performance Isolation for Concurrent Deep Learning Workloads},
  author={Zhao, Wei and Jayarajan, Anand and Pekhimenko, Gennady},
  booktitle={Proceedings of the 30th ACM International Conference on Architectural Support for Programming Languages and Operating Systems, Volume 1}, pages={1052--1068},
  publisher={Association for Computing Machinery}, address={New York, NY, USA}, year={2025}, doi={10.1145/3669940.3707282}
}

@inproceedings{agrawal2024dma,
  title={{ConCCL}: Optimizing {ML} Concurrent Computation and Communication with {GPU DMA} Engines},
  author={Agrawal, Anirudha and Aga, Shaizeen and Pati, Suchita and Islam, Mahzabeen},
  booktitle={2025 IEEE International Symposium on Performance Analysis of Systems and Software (ISPASS)}, pages={1--11},
  publisher={IEEE}, address={Piscataway, NJ, USA}, year={2025}, doi={10.1109/ISPASS64960.2025.00018}
}

@inproceedings{hu2025demystifying,
  title={Demystifying {NCCL}: An In-Depth Analysis of {GPU} Communication Protocols and Algorithms},
  author={Hu, Zhiyi and Shen, Siyuan and Bonato, Tommaso and Jeaugey, Sylvain and Alexander, Cedell and Spada, Eric and Dinan, James and Hammond, Jeff and Hoefler, Torsten},
  booktitle={2025 IEEE Symposium on High-Performance Interconnects (HOTI)}, pages={48--59},
  publisher={IEEE}, address={Piscataway, NJ, USA}, year={2025}, doi={10.1109/HOTI66940.2025.00024}
}

@article{agostini2018gpudirect,
  title={{GPUDirect Async}: Exploring {GPU} Synchronous Communication Techniques for {InfiniBand} Clusters},
  author={Agostini, Elena and Rossetti, Davide and Potluri, Sreeram},
  journal={Journal of Parallel and Distributed Computing}, volume={114}, pages={28--45}, year={2018}, doi={10.1016/j.jpdc.2017.12.007}
}

@article{unat2024landscape,
  title={The Landscape of {GPU}-Centric Communication},
  author={Unat, Didem and Turimbetov, Ilyas and Issa, Mohammed Kefah Taha and Sa\u{g}bili, Do\u{g}an and Vella, Flavio and De Sensi, Daniele and Ismayilov, Ismayil},
  journal={ACM Computing Surveys}, volume={58}, number={12}, pages={1--36}, year={2026}, doi={10.1145/3813799}
}

@inproceedings{wang2020blink,
  title={Blink: Fast and Generic Collectives for Distributed {ML}},
  author={Wang, Guanhua and Venkataraman, Shivaram and Phanishayee, Amar and Devanur, Nikhil and Thelin, Jorgen and Stoica, Ion},
  booktitle={Proceedings of the 3rd Conference on Machine Learning and Systems (MLSys)}, numpages={15},
  publisher={MLSys}, address={Austin, TX, USA}, year={2020},
  url={https://proceedings.mlsys.org/paper_files/paper/2020/hash/cd3a9a55f7f3723133fa4a13628cdf03-Abstract.html}
}

@misc{nvidia2025cudaguide,
  title={{CUDA C++ Programming Guide}: Thread Block Clusters}, author={{NVIDIA Corporation}}, year={2026},
  note={CUDA Programming Guide, accessed October 25, 2026},
  howpublished={\url{https://docs.nvidia.com/cuda/cuda-programming-guide/01-introduction/programming-model.html}}
}

@misc{nvidia2023matmul,
  title={Matrix Multiplication Background User's Guide}, author={{NVIDIA Corporation}}, year={2023},
  note={NVIDIA Deep Learning Performance Documentation, Section 3.2, Wave Quantization},
  howpublished={\url{https://docs.nvidia.com/deeplearning/performance/dl-performance-matrix-multiplication/index.html}}
}

\flushcolsend
\clearpage
\raggedcolsend
\section*{Appendix}
\suppressfloats[t]
\appendix
\setcounter{subsection}{0}
\renewcommand{\thesubsection}{\Alph{subsection}}
\crefalias{subsection}{appendix}

\subsection{Full Collective Measurements}
\label{app:microbench}

This appendix lists the operator measurements behind
\Cref{sec:wave-quant,sec:eval-attribution,sec:eval-scaling}: the
reduce-scatter overlap experiment behind \Cref{sec:wave-quant}
(\Cref{tab:eval-overlap}, plotted in \Cref{fig:eval-cta}), the SM-id
probe of CTA placement (\Cref{tab:smid}), the
training-shape replay of \Cref{fig:eval-mechanism}
(\Cref{tab:eval-shape-replay}), and all-gather and reduce-scatter
latency and buffer memory across 1--8 nodes
(\Cref{tab:eval-ag,tab:eval-rs}).

\begin{figure}[ht]
  \centering
  \includegraphics[width=\columnwidth]{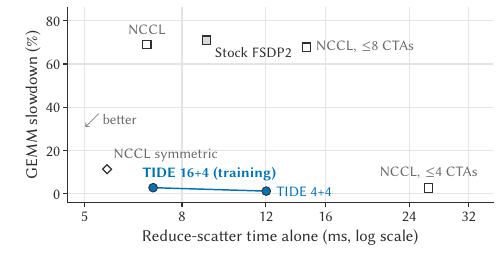}
  \caption{Reduce-scatter time versus one-wave GEMM slowdown
  (\Cref{tab:eval-overlap}). Hollow markers: pre-packed \dtype{fp32}
  input.}
  \Description{Scatter plot of reduce-scatter time alone against the
  slowdown of a concurrent one-wave GEMM. Stock FSDP2 and NCCL
  configurations with default or eight CTAs sit near 68 to 71 percent
  slowdown, and NCCL's symmetric kernel at 11 percent; NCCL capped at
  four CTAs reaches 2.7 percent but takes 26.4 ms; TIDE with 16 or 4
  sum CTAs sits at 1.3 to 2.9 percent and 7.0 to 12.0 ms.}
  \label{fig:eval-cta}
\end{figure}

\begin{table}[ht]
  \centering
  \caption{Concurrent 32~MiB reduce-scatters on 8 nodes at a fixed SM
  clock (medians of three launches). Bold: lowest slowdown among rows
  no slower than Stock FSDP2.}
  \label{tab:eval-overlap}
  \footnotesize
  \setlength{\tabcolsep}{3pt}
  \begin{tabular*}{\columnwidth}{@{\extracolsep{\fill}}lcrrrr@{}}
    \toprule
    & & \textbf{RS alone} & \multicolumn{3}{c}{\textbf{Victim slowdown (\%)}} \\
    \cmidrule(l){4-6}
    \textbf{Configuration} & \textbf{CTAs} & \textbf{(ms)}
    & \textbf{1-wave} & \textbf{28-wave} & \textbf{Elementwise} \\
    \midrule
    \system, training & $16+4$ & 6.96 & \textbf{2.9} & \textbf{4.2} & \textbf{1.6} \\
    \system, 4+4 & $4+4$ & 12.02 & 1.3 & 4.0 & 0.6 \\
    Stock FSDP2 & $16$ & 9.01 & 71.0 & 13.6 & 13.3 \\
    FSDP2+SymmMem & $13$ & 7.45 & 9.7 & 10.6 & 13.6 \\
  \midrule
    NCCL default & $16$ & 6.75 & 69.0 & 13.5 & 7.4 \\
    NCCL symmetric & $13$ & 5.57 & 11.4 & 10.2 & 13.4 \\
    NCCL, $\le8$ CTAs & $\le8$ & 14.61 & 67.8 & 7.0 & 7.3 \\
    NCCL, $\le4$ CTAs & $\le4$ & 26.35 & 2.7 & 3.8 & 4.1 \\
    \bottomrule
  \end{tabular*}\par\vspace{3pt}
  {\footnotesize\raggedright Victims: the one-wave GEMM of
  \Cref{sec:wave-quant}, a GEMM of about 28 waves, and a 256~MiB
elementwise kernel. NCCL rows start from pre-packed \dtype{fp32}
input.\par}
\end{table}

\begin{table}[ht]
  \centering
  \caption{SM-id probe during a 32~MiB reduce-scatter on 8 nodes at a
  fixed SM clock (two launches). The probe launches the one-wave
  GEMM's 32 four-CTA clusters, each CTA filling an SM and recording
  its SM id and entry and exit timestamps; \emph{displaced} counts
  probe CTAs that start after the first wave. Resident CTAs come from
  a second probe with one CTA on each of the 148 SMs.}
  \label{tab:smid}
  \footnotesize
  \setlength{\tabcolsep}{4pt}
  \begin{tabular*}{\columnwidth}{@{\extracolsep{\fill}}lrrr@{}}
    \toprule
    \textbf{Configuration} & \textbf{Displaced} & \textbf{Windows} & \textbf{Resident} \\
     & \textbf{CTAs (mode)} & \textbf{hit (\%)} & \textbf{CTAs (mean)} \\
    \midrule
    \system, training & 0 & 0 & 2.0 \\
    \system, 4+4 & 0 & 0 & 1.5 \\
    Stock FSDP2 & 12 & 65 & 13.4 \\
    FSDP2+SymmMem & 0 & 0 & 12.0 \\
  \midrule
    NCCL default & 12 & 64 & 14.7 \\
    NCCL symmetric & 0 & 7 & 12.6 \\
    NCCL, $\le8$ CTAs & 4 & 63 & 7.3 \\
    NCCL, $\le4$ CTAs & 0 & 0 & 3.5 \\
    \bottomrule
  \end{tabular*}
\end{table}

\begin{table}[ht]
  \centering
  \caption{Training-shape replay at a fixed SM clock (medians of three
  launches of 30 isolated calls).}
  \label{tab:eval-shape-replay}
  \footnotesize
  \setlength{\tabcolsep}{4.0pt}
  \begin{tabular*}{\columnwidth}{@{\extracolsep{\fill}}llrrr@{}}
    \toprule
    & & \textbf{Latency} & \textbf{Victim} & \textbf{SM-ms} \\
    \textbf{Op} & \textbf{System} & \textbf{p50 (ms)}
    & \textbf{slowdown} & \textbf{per call} \\
    \midrule
    \multicolumn{5}{@{}l}{\emph{Llama-3.1 70B}} \\
    AG & Stock FSDP2   & 8.07 & 6.60\% & 218.2 \\
    AG & FSDP2+SymmMem & 7.23 & 3.05\% & 86.3 \\
    AG & \system & \textbf{4.90} & \textbf{0.33\%} & \textbf{0.0} \\
    RS & Stock FSDP2   & 14.12 & 11.61\% & 651.4 \\
    RS & FSDP2+SymmMem & 11.97 & 9.60\% & 584.3 \\
    RS & \system & \textbf{10.60} & \textbf{3.48\%} & \textbf{44.8} \\
    \midrule
    \multicolumn{5}{@{}l}{\emph{Wan2.2}} \\
    AG & Stock FSDP2   & 3.88 & 7.97\% & 104.6 \\
    AG & FSDP2+SymmMem & 3.19 & 2.66\% & 38.3 \\
    AG & \system & \textbf{2.56} & \textbf{0.48\%} & \textbf{0.0} \\
    RS & Stock FSDP2   & 6.12 & 23.17\% & 265.9 \\
    RS & FSDP2+SymmMem & \textbf{5.24} & 20.63\% & 240.1 \\
    RS & \system & 5.58 & \textbf{2.05\%} & \textbf{19.2} \\
    \bottomrule
  \end{tabular*}
\end{table}

\begin{table*}[t]
  \centering
  \caption{All-gather latency, min (median) ms, for $P=30$ tensors.
  ``+ Copy-out'' adds FSDP2's fused \code{split\_with\_sizes\_copy};
  \system lands directly. Bold: the fastest path that yields FSDP2's
  parameter layout; the last column compares medians.}
  \label{tab:eval-ag}
  \footnotesize
  \setlength{\tabcolsep}{3pt}
  \begin{tabular*}{\textwidth}{@{\extracolsep{\fill}}rrcccccc@{}}
    \toprule
    \textbf{Nodes (GPUs)} & \textbf{Shard} & \textbf{\system-AG (0 SM)}
    & \textbf{NCCL Ring} & \textbf{NCCL CE} & \textbf{Ring + Copy-out}
    & \textbf{CE + Copy-out} & \textbf{vs.\ Best FSDP2} \\
    \midrule
    1 (8 GPUs)  & 32 MiB  & \textbf{0.38 (0.38)} & 0.40 (0.41) & 0.35 (0.36) & 0.50 (0.51) & 0.45 (0.46) & \textbf{1.22$\times$} \\
    2 (16 GPUs) & 32 MiB  & \textbf{1.11 (1.12)} & 1.35 (1.59) & 1.08 (1.10) & 1.53 (1.58) & 1.26 (1.28) & \textbf{1.14$\times$} \\
    4 (32 GPUs) & 32 MiB  & \textbf{2.57 (2.60)} & 2.71 (2.79) & 3.14 (3.22) & 3.11 (3.14) & 3.56 (3.62) & \textbf{1.21$\times$} \\
    8 (64 GPUs) & 32 MiB  & \textbf{5.52 (5.86)} & 5.46 (5.78) & 7.25 (7.42) & 6.15 (6.21) & 7.96 (8.19) & \textbf{1.06$\times$} \\
    \midrule
    1 (8 GPUs)  & 128 MiB & \textbf{1.40 (1.41)} & 1.50 (1.52) & 1.27 (1.28) & 1.85 (1.87) & 1.61 (1.63) & \textbf{1.16$\times$} \\
    2 (16 GPUs) & 128 MiB & \textbf{4.19 (4.20)} & 5.29 (6.09) & 4.04 (4.13) & 5.99 (6.84) & 4.73 (4.76) & \textbf{1.13$\times$} \\
    4 (32 GPUs) & 128 MiB & \textbf{9.76 (9.80)} & 10.78 (13.48) & 12.03 (13.27) & 12.42 (13.01) & 13.64 (14.06) & \textbf{1.33$\times$} \\
    8 (64 GPUs) & 128 MiB & \textbf{20.89 (21.41)} & 21.73 (23.41) & 28.01 (30.90) & 24.74 (30.35) & 31.00 (32.75) & \textbf{1.42$\times$} \\
    \bottomrule
  \end{tabular*}
\end{table*}

\begin{table*}[t]
  \centering
  \caption{Reduce-scatter latency, min (median) ms, and buffer
  (GiB/GPU) for $P=16$ \dtype{bf16} gradients. FSDP2 columns include
  \code{\_chunk\_cat}; Bare NCCL starts from packed \dtype{fp32}.
  Bold: lowest among per-parameter paths, and the smaller buffer.}
  \label{tab:eval-rs}
  \footnotesize
  \setlength{\tabcolsep}{3pt}
  \begin{tabular*}{\textwidth}{@{\extracolsep{\fill}}rrccccc@{}}
    \toprule
    \textbf{Nodes (GPUs)} & \textbf{Output} & \textbf{\system-RS} &
    \textbf{Stock FSDP2} & \textbf{FSDP2+SymmMem} & \textbf{Bare NCCL} &
    \textbf{Buffer (\system\,/\,FSDP2)} \\
    \midrule
    1 (8 GPUs)  & 32 MiB  & \textbf{0.41 (0.42)} & 0.59 (0.60) & 0.60 (0.61) & 0.39 (0.40) & 0.29 / 0.28 \\
    2 (16 GPUs) & 32 MiB  & 2.08 (2.12) & 1.67 (1.75) & \textbf{1.16 (1.27)} & 1.30 (1.45) & 0.58 / 0.53 \\
    4 (32 GPUs) & 32 MiB  & 3.70 (3.75) & 3.39 (4.06) & \textbf{3.06 (3.07)} & 2.69 (2.78) & \textbf{0.61 / 1.03} \\
    8 (64 GPUs) & 32 MiB  & \textbf{6.61} (8.12) & 6.82 (7.96) & 6.67 (\textbf{6.84}) & 5.40 (5.52) & \textbf{0.64 / 2.03} \\
    \midrule
    1 (8 GPUs)  & 128 MiB & \textbf{1.29 (1.33)} & 2.23 (2.25) & 2.36 (2.37) & 1.44 (1.46) & \textbf{0.38 / 1.12} \\
    2 (16 GPUs) & 128 MiB & 7.63 (7.68) & 6.70 (6.78) & \textbf{4.62 (4.96)} & 5.17 (5.27) & 2.26 / 2.12 \\
    4 (32 GPUs) & 128 MiB & 13.93 (13.99) & 13.67 (14.02) & \textbf{11.64 (11.70)} & 10.66 (11.30) & \textbf{2.38 / 4.12} \\
    8 (64 GPUs) & 128 MiB & \textbf{25.02 (25.59)} & 27.69 (34.77) & 26.42 (27.16) & 21.62 (23.91) & \textbf{2.51 / 8.12} \\
    \bottomrule
  \end{tabular*}
\end{table*}

\subsection{Scheduler Characterization}
\label{app:clc}

A statically scheduled persistent GEMM launches $N_{\mathrm{SM}}$
long-running CTAs that loop over tiles in software, so a
higher-priority communication kernel must wait for one of them to
retire. Dynamic persistent GEMMs based on Cluster Launch Control
(CLC)~\cite{ptx9}, in contrast, request further work from the hardware
scheduler at tile boundaries. \Cref{tab:clc-scheduling} characterizes
both with two probes during a running persistent GEMM. A high-priority
communication kernel enters a CLC GEMM within two tile boundaries,
whereas static or equal-priority scheduling delays it until the final
GEMM wave. At low communication duty (one reduce-scatter every 16
GEMMs), cooperative CLC sharing slows the GEMM by $1.4\%$ versus $4.9\%$
with 16 statically reserved SMs under the power cap, and by $1.2\%$
versus $13.7\%$ at a fixed clock. This is cooperative scheduling, not
preemption: it favors kernels with few short-lived CTAs, such as
\system's sum and landing kernels, but statically scheduled persistent
GEMMs remain non-preemptible.

\begin{table}[t]
  \centering
  \caption{Scheduler probes during a persistent GEMM. CLC-normal/high:
  normal/high-priority communication.}
  \label{tab:clc-scheduling}
  \footnotesize
  \setlength{\tabcolsep}{3pt}
  \begin{tabular*}{\columnwidth}{@{\extracolsep{\fill}}lcc@{}}
    \toprule
    \textbf{Metric} & \textbf{Static or CLC-normal} & \textbf{CLC-high} \\
    \midrule
    Communication admission & Final GEMM wave & $\le 2$ tile boundaries \\
    Added RS latency & $0.38\text{--}0.90$ ms & $0.09\text{--}0.54$ ms \\
    GEMM penalty & $\approx 0$ (RS waits) & $0.015\text{--}0.14$ ms \\
    \bottomrule
  \end{tabular*}\par\vspace{3pt}
  {\footnotesize\raggedright Admission uses a synthetic communication
  kernel; latency (relative to no GEMM) and penalty use 64--512~MiB NCCL
  reduce-scatters submitted 30\% into the GEMM, across power-capped and
  fixed-clock runs.\par}
\end{table}

\subsection{Convergence}
\label{app:convergence}

\begin{figure}[t]
  \centering
  \includegraphics[width=\columnwidth]{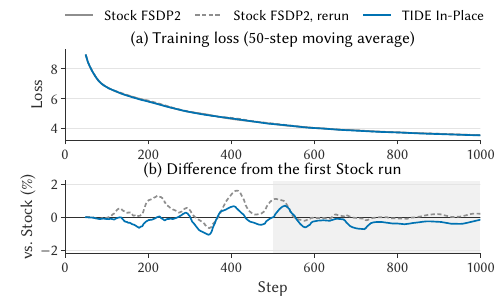}
  \caption{Qwen3-30B-A3B trained from scratch on public C4 (64 GPUs,
  1,000 steps). (a) 50-step moving-average loss. (b) Difference from the
  first Stock run; \Cref{sec:eval-correctness} reports maxima over the
  shaded steps 500--1,000.}
  \Description{Two panels over 1,000 training steps. Top: moving-average
  loss of two Stock FSDP2 runs and TIDE In-Place, which overlap as the
  loss falls from about 6 to 3.5. Bottom: percent difference from the
  first Stock run; after step 500, TIDE stays within 0.73 percent and the
  second Stock run within 1.12 percent.}
  \label{fig:convergence}
\end{figure}

The run in \Cref{fig:convergence} trains Qwen3-30B-A3B (128 experts,
eight active per token) from random initialization with four
4,096-token sequences per GPU. \system uses In-Place; the two Stock
FSDP2 runs use the same configuration, so their difference measures
ordinary run-to-run variation.

\subsection{Latin-Square Collision-Freedom Proof}
\label{app:latin-proof}

\begin{proof}[Proof of \Cref{prop:latin}]
Let $c_0$ be the first contributor in wave $w$ and
$\mathit{len}$ that wave's contributor count.
Within any wave $w \in [0, W)$, the contributor indices form a contiguous
integer interval $[c_0, c_0 + \mathit{len})$ of length
$\mathit{len} \le W_{\mathrm{wave}} \le L$. For all remote contributors
$c \in \{1, \dots, M-1\}$, $c \bmod M = c$. If wave $w$ contains only
remote contributors, their $(c \bmod M)$ values are $\mathit{len} \le L$
consecutive integers, which are pairwise incongruent modulo $L$; thus
$j(s, c) = (s - c) \bmod L$ is injective over wave $w$. The only
contributor with $c \bmod M = 0$ is the own node $c = M$, which appears
exclusively as the final element of the last wave $w = W - 1$ alongside
remote contributors $c \in [(W-1)W_{\mathrm{wave}} + 1, M - 1]$. Because
$W_{\mathrm{wave}} \mid L$ (say $L = q W_{\mathrm{wave}}$ for integer
$q \ge 1$), the open interval $((W-1)W_{\mathrm{wave}}, W \cdot W_{\mathrm{wave}})$
contains no multiple of $W_{\mathrm{wave}}$, and hence no multiple of
$L = q W_{\mathrm{wave}}$. Therefore, every remote contributor in the
final wave satisfies $c \bmod L \ne 0 = (M \bmod M) \bmod L$, and
$j(s, c)$ is strictly injective across every wave $w$.
\end{proof}

For non-divisor test widths $W_{\mathrm{wave}} \nmid L$, the final wave
can straddle a multiple of $L$ alongside the own node $c = M$ (e.g., at
$M = 9, L = 8, W_{\mathrm{wave}} = 3$, wave $w = 2$ holds
$c \in \{7, 8, 9\}$, and $c = 8$ and $c = 9$ both have
$(c \bmod 9) \bmod 8 = 0$), so the planner checks the per-wave inbox
block bitmask and rejects such configurations. An offline regression
test covers $M\in[2,1024]$ at $L=8$: all $458{,}304$ accepted base plans
$(M,W_{\mathrm{wave}},k,me)$, where $me$ selects the memory-efficient
variant, are collision-free and stale-block free, and $65{,}472$
colliding plans are rejected; sub-slice pass counts
$n_{\mathrm{part}}\in\{2,3,4\}$ are checked separately for every
eligible base plan.

\subsection{Pipeline Protocol Details}
\label{app:pipeline-details}

\begin{figure}[t]
  \centering
\begin{tikzpicture}[td, x=1pt, y=1pt]
\def\tdX#1{40+0.9*#1}
\tikzset{
  tdcmp/.style={draw=tdINK!45, line width=0.3pt, fill=tdCMP},
  tdrd/.style={draw=tdSMd, line width=0.35pt, fill=tdSMf},
}
\node[font=\tdsmall, anchor=east, inner sep=1pt] at (36,-36.5) {host};
\node[font=\tdsmall, anchor=east, inner sep=1pt] at (36,-96.5) {GPU};
\foreach \tdt/\tdu/\tdlab in {0/30/$\ell$, 30/60/$\ell{-}1$, 60/90/$\ell{-}2$}{%
  \filldraw[draw=tdINK!60, line width=0.35pt, fill=white] ({\tdX{\tdt}},-33) rectangle ({\tdX{\tdu}},-40);
  \node[font=\tdtiny, inner sep=0pt] at ({\tdX{(\tdt+\tdu)/2}},-36.5) {\tdlab};}
\node[font=\tdtiny, text=tdMUTE, anchor=west, inner sep=2.5pt] at ({\tdX{90}},-36.5) {enqueue};
\foreach \tdt in {36,66}{%
  \filldraw[draw=tdTAX, fill=tdTAXf, line width=0.4pt] ({\tdX{\tdt}},-29.5) circle (1.5pt);}
\node[font=\tdtiny, text=tdTAX, anchor=south, inner sep=1pt] at ({\tdX{51}},-27.5)
  {allocate grads of $\ell{-}1$, $\ell{-}2$};
\draw[draw=tdINK, line width=0.8pt] ({\tdX{58}},-31) -- ({\tdX{58}},-42);
\node[font=\tdtiny, anchor=north west, inner sep=0pt, align=left] at ({\tdX{0}},-44.5)
  {stream-ordered release at\\post\_bwd: $\ell{-}2$ reuses it};
\filldraw[tdcmp] ({\tdX{6}},-93) rectangle ({\tdX{104}},-100);
\node[font=\tdtiny, inner sep=0pt] at ({\tdX{55}},-96.5) {bwd $\ell$};
\filldraw[tdcmp] ({\tdX{104}},-93) rectangle ({\tdX{200}},-100);
\node[font=\tdtiny, inner sep=0pt] at ({\tdX{175}},-96.5) {bwd $\ell{-}1$};
\filldraw[tdrd] ({\tdX{104}},-102.5) rectangle ({\tdX{146}},-108);
\node[font=\tdtiny, inner sep=0pt] at ({\tdX{125}},-105.25) {read $\ell$};
\draw[draw=tdTAX, line width=0.5pt, dash pattern=on 1.5pt off 1pt]
  ({\tdX{146}},-108.5) -- ({\tdX{146}},-29.5);
\node[font=\tdtiny, text=tdTAX, anchor=west, inner sep=1.5pt, align=left] at ({\tdX{146}},-62)
  {GPU-safe release:\\too late, new segments};
\draw[draw=tdTAX, line width=0.45pt, ->] ({\tdX{146}},-29.5) -- ({\tdX{66}+2.2},-29.5);
\end{tikzpicture}
  \caption{Host-time allocation: the host allocates later layers'
  gradients before the GPU finishes reading layer $\ell$, so a GPU-safe
  release is too late for them to reuse its storage.}
  \Description{A host timeline that enqueues three backward layers and
  allocates gradients early, above a GPU timeline whose read of layer
  ell finishes later; the GPU-safe release point lies after the host
  allocations.}
  \label{fig:host-time}
\end{figure}
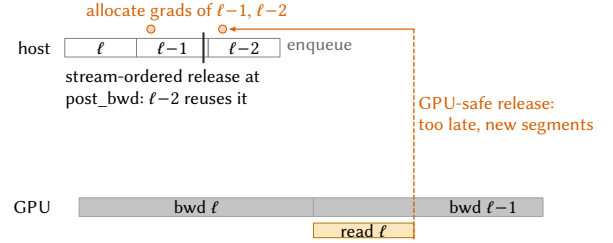

\paragraph{Rejected release schedules.}
Three alternatives motivate the split used in \Cref{sec:two-phase}.
First, returning read completion while consecutive calls shared one
outbox bank tied the read event to the previous call's network phase and
lengthened compute-stream waits. Second, forcing the compute stream to
wait at a mid-layer release point exposed local-read latency on the
critical path. Third, \code{record\_stream} and event-polling release
preserved GPU safety but occurred too late in host time
(\Cref{fig:host-time}); peak reserved memory grew because allocations
for two later layers had already been enqueued.

\paragraph{Detailed event protocol.}
Let $D_{\mathrm{rem}}=n_{\mathrm{part}}(M-1)$ be the number of remote
destination blocks across all passes. \Cref{alg:pipelined-rs} gives the
full schedule; a missing prefetch gate is treated as an already-complete
event.
Here $J$ indexes remote destination blocks, $\code{outbox}[b]$ is bank
$b$'s reusable send storage, and $s_{\mathrm{comm}}$ is the
communication stream. Events $\code{readyEv}$, $\code{gateEv}$, and
$\code{readEv}$ signal source availability, the prefetch gate, and
source-read completion; $\code{sumEv}[b][J]$ signals local reduction for
destination $J$.

\begin{algorithm}[t]
\caption{Two-Phase Pipelined \system-RS (Epoch $e$, Bank $b=e\bmod 2$)}
\label{alg:pipelined-rs}
\footnotesize
\begin{algorithmic}[1]
\Require Gradients $\{G_p\}$, output $\mathrm{out}$, events
  $\code{readyEv},\code{gateEv},\code{readEv}$, stream $s_{\mathrm{comm}}$
\Statex \textbf{Read phase} (waits only on $\code{readyEv}$)
\For{$J=1\dots\min(n_{\mathrm{out}},D_{\mathrm{rem}})$}
  \State wait until $\code{outbox}[b][J\bmod n_{\mathrm{out}}]$ is reusable
  \State push and sum destination $J$; record $\code{sumEv}[b][J]$
\EndFor
\State sum own-node partials into bank-$b$ spares
\If{$n_{\mathrm{out}}\ge D_{\mathrm{rem}}$} record $\code{readEv}$ \EndIf
\Statex \textbf{Network phase} (behind call $e-1$ and $\code{gateEv}$)
\For{wave $w=0\dots W-1$ with buffer generation $q=w\bmod 2$}
  \For{each $k$-way round in wave $w$}
    \State wait for sums and credits; issue striped puts; wait for arrivals
    \State mark sent outbox blocks reusable
    \If{an unsummed destination remains}
      \State sum it into the newly reusable block; record final
        $\code{readEv}$ when all sources are consumed
    \EndIf
  \EndFor
  \State copy the own-node spare if present; land wave $w$; mark buffer generation $q$ reusable
\EndFor
\State signal credits and record network-order completion
\end{algorithmic}
\end{algorithm}

\paragraph{Protocol invariants.}
The implementation checks four invariants with fault-injection negative
tests: (I1) every contributor lands exactly once and stale blocks from
the previous use of the same buffer generation are zeroed; (I2) gradients or slot copies remain live through
\code{readEv}; (I3) staging, outbox, and inbox storage is reused only
after its previous consumer completes; and (I4) every event or memory-op
wait has an earlier producer in the common collective order.

\paragraph{CUDA hardware work queues.}
With the default eight CUDA connections, blocking memory operations from
the network phase aliased onto the same hardware FIFO as later local
reads and stretched the read span; raising
\code{CUDA\_\allowbreak DEVICE\_\allowbreak MAX\_\allowbreak CONNECTIONS}
to 32 removed the observed aliasing. This is an empirical requirement:
CUDA does not guarantee each stream its own hardware queue.

\subsection{FSDP2 Adapter and Expert Subgroups}
\label{app:adapter-details}

\paragraph{Position-aware ring slots and ring-tail folding.}
Under FSDP2's default prefetch depth of 1, at most two ring groups are
unsharded at once, so group $g$ (in execution order) uses all-gather
slot $g \bmod 2$, and \system sizes each slot for the groups assigned
to it rather than for the global maximum. The final Transformer block
(\code{reshard\_\allowbreak
after\_\allowbreak forward=False}) is live only across the turn-around
boundary and alternates with its predecessor like a ring group; \system
folds it into the 2-slot ring and re-shards it after forward during
no-gradient validation, saving a dedicated full-size slot.

\paragraph{Zero-offset views and PyTorch \code{MemPool} accounting.}
FSDP2 builds unsharded parameters with \code{torch.\allowbreak as\_strided}
at \code{storage\_\allowbreak offset=0}, and EMA validation's
\code{copy.\allowbreak deepcopy} of \code{get\_\allowbreak model\_\allowbreak
state\_\allowbreak dict()} raises \code{invalid python storage} if
a \code{DTensor}'s local tensor has a non-zero storage offset. \system's
C++ adapter therefore roots a fresh \code{at::Storage} at the arena
offset, with
\code{storage\_\allowbreak offset() == 0} and its own version counter. To
make symmetric arenas visible to \code{torch.\allowbreak cuda.\allowbreak
memory\_\allowbreak reserved()}, \system allocates the arena inside a
\code{torch.\allowbreak cuda.\allowbreak MemPool}, registers the
2~MiB-aligned pointer with NCCL, and avoids a reference cycle.

\paragraph{Expert-parallel ($L_{\mathrm{sub}} = 1$) rail lanes and HSDP.}
When MoE routed experts are sharded across $E_{\mathrm{ep}} = L = 8$
modulo subgroups of size $N / E_{\mathrm{ep}}$, each subgroup has
$L_{\mathrm{sub}} = 1$ rank per node ($M$ rail peers, no LSA peers),
and \system specializes both collectives on a dedicated lane.
All-gather runs $\lceil \log_2 M \rceil$
recursive-doubling RDMA rounds directly from landing to landing, with no
receive-stage memory or copy-engine scatters. Reduce-scatter skips
Level~1 staging, sends \dtype{bf16} pieces across four rail parts into a
double-buffered inbox, and uses a fixed eight-source TMA sum to combine the
arriving blocks and own contribution in \dtype{fp32}. Under HSDP, the
shard-group output proceeds to FSDP2's replica all-reduce.

\subsection{Device-Initiated Path and Sizing Solver}
\label{app:gin-solver}

\paragraph{Device-initiated network phase.}
When GIN runs over GDAKI, the optional device path runs the network
phase as one 1-CTA kernel per wave plus a 1-CTA finalizer on a
highest-priority stream. Before launching the kernel for wave $w$, the
host queues waits for that wave's Level~1 sums and, for $w \ge 2$, for
the landing of wave $w-2$, which last used the same inbox buffer generation. In
each round, the kernel waits for per-peer credits, issues striped
\code{gin.put} calls with a strong signal on the last chunk per peer
and RMA context, waits for incoming data signals, and flushes its
contexts before an outbox block is reused. Waits are timed against
\code{\%globaltimer} because SM cycle counters stretch under clock
throttling. On timeout, the kernel records the NIC send-queue state in
mapped host memory and parks in a loop of volatile reads, which the
compiler cannot remove as a side-effect-free infinite loop; its
completion event therefore never fires, and the host monitor exits
fail-stop.

\paragraph{RoCE traffic-class alignment.}
GDAKI queue pairs do not inherit the NIC's global traffic class and
otherwise use a lossy priority that can stall under congestion, so
\system sets \code{NCCL\_GIN\_IB\_TC} to the traffic class of the host
RMA path before creating communicators.

\paragraph{Analytical sizing solver equations.}
With input itemsize $b_{\mathrm{grad}} \in \{2, 4\}$ and wire time
$t_{\mathrm{NIC}} = (M - 1) S_{\mathrm{fp32}} / \mathrm{BW}_{\mathrm{NIC}}$,
the first-order solver selects $C_{\mathrm{sum}}$ and
$C_{\mathrm{mc}}$. Here $\mathrm{BW}_{\mathrm{NIC}}$ is sustained
per-rail throughput, $C_{\mathrm{sum}}$ and $C_{\mathrm{mc}}$ are sum
and landing CTA budgets, $\mathrm{BW}_{\mathrm{sum}}$ and
$\mathrm{BW}_{\mathrm{land}}$ their per-CTA rates, and
$\mathrm{BW}_{\mathrm{cap}}$ the landing fabric ceiling:
\begin{align}
  C_{\mathrm{sum}} &= \left\lceil \gamma_{\mathrm{skew}} \cdot
  \left(\frac{L \, b_{\mathrm{grad}}}{4} + 1\right)
  \frac{\mathrm{BW}_{\mathrm{NIC}}}{\mathrm{BW}_{\mathrm{sum}}} \right\rceil, \label{eq:solver-csum} \\
  C_{\mathrm{mc}} &= \min\!\left(\left\lfloor
  \frac{\mathrm{BW}_{\mathrm{cap}}}{\mathrm{BW}_{\mathrm{land}}} \right\rfloor,\,
  \left\lceil \frac{S_{\mathrm{fp32}} / n_{\mathrm{part}}}{\eta_{\mathrm{tail}} \cdot t_{\mathrm{NIC}} \cdot \mathrm{BW}_{\mathrm{land}}} \right\rceil\right), \label{eq:solver-cmc}
\end{align}
where $\gamma_{\mathrm{skew}} = 1.5$ absorbs arrival skew and
$\eta_{\mathrm{tail}} = 0.05$ bounds the final-wave landing tail to $5\%$
of $t_{\mathrm{NIC}}$. Let
$N_{\mathrm{rounds}}(M-1,k,W_{\mathrm{wave}})$ count network rounds for
$M-1$ contributions at wave width $W_{\mathrm{wave}}$ and at most $k$
transfers per round; $t_{\mathrm{rnd}}$ is calibrated wait overhead.
Round width $k$ starts at $\min(2, M-1)$ and
doubles up to $k_{\max} = 4$ whenever $N_{\mathrm{rounds}}(M-1, k,
W_{\mathrm{wave}}) \cdot t_{\mathrm{rnd}} > 0.10 \cdot t_{\mathrm{NIC}}$
and doubling $k$ reduces the round count. On the evaluation platform,
calibration gives
$\mathrm{BW}_{\mathrm{NIC}}=47.5\text{ GB/s}$,
$\mathrm{BW}_{\mathrm{sum}}=99\text{ GB/s/CTA}$,
$\mathrm{BW}_{\mathrm{land}}=21\text{ GB/s/CTA}$ in situ with a
$90\text{ GB/s}$ fabric ceiling, and
$t_{\mathrm{rnd}}=1.3\text{ ms}$.

\end{document}